\documentclass[lettersize,journal]{IEEEtran}
\usepackage{amsmath,amsfonts}
\usepackage{algorithmic}
\usepackage{subfigure}
\usepackage{float} 
\usepackage{amssymb}
\usepackage{hyperref}
\usepackage[normalem]{ulem}
\usepackage[table]{xcolor}
\usepackage{colortbl}
\definecolor{groupgray}{RGB}{242,244,247}
\definecolor{oursblue}{RGB}{235,242,252}
\definecolor{accgray}{RGB}{248,248,248}
\useunder{\uline}{\ul}{}
\usepackage{multirow}
\usepackage{algorithm}
\usepackage{comment}
\usepackage{makecell}
\usepackage{booktabs}
\usepackage{enumitem}
\usepackage{array}
\usepackage{textcomp}
\usepackage{stfloats}
\usepackage{url}
\usepackage{verbatim}
\usepackage{graphicx}
\usepackage{cite}
\usepackage{xcolor}
\usepackage{balance}
\usepackage{amsmath, amsthm, amssymb}

\newtheorem{definition}{Definition}
\newtheorem{proposition}{Proposition}

\begin{document}

\title{Delayed Commitment for Representation Readiness in Stage-wise Audio-Visual Learning}

\author{Xinmeng Xu, Haoran Xie, S. Joe Qin, Lin Li, Xiaohui Tao, Fu Lee Wang
\thanks{This work was supported by the Research Impact Fund by the Research Grants Council of Hong Kong (Project No. 130272); two grants from the Research Grants Council of the Hong Kong Special Administrative Region, China (R1015-23 and UGC/FDS16/E17/23); the Faculty Research Grants (SDS24A8, SDS25A15 and SDS24A19), Interdisciplinary \& Strategic Research Grant (ISRG252606), and the Direct Grants (DR25E8 and DR26F2) of Lingnan University, Hong Kong.}\\
\thanks{Xinmeng Xu, Haoran Xie, and S. Joe Qin are with the Division of Artificial Intelligence, Lingnan University, Tuen Mun, Hong Kong SAR (e-mail: xinmengxu@ln.edu.hk; hrxie@ln.edu.hk; joeqin@ln.edu.hk).\\
\indent Lin Li is with the School of Computer Science and Artificial Intelligence,  Wuhan University of Technology, Wuhan, China (email: cathylilin@whut.edu.cn).\\
\indent Xiaohui Tao is with the School of Mathematics, Physics and Computing, University of Southern Queensland, Toowoomba, Australia (email:
xtao@usq.edu.au).\\
\indent Fu Lee Wang is with the School of Science and Technology, Hong Kong
Metropolitan University, Ho Man Tin, Hong Kong SAR (email: pwang@hkmu.edu.hk).}
\thanks{(\textit{Corresponding author: Haoran Xie.})}
}



\maketitle
\begin{abstract}
Stage-wise audio-visual encoders propagate fused intermediate states across layers, making the formation of later representations depend on the readiness of earlier fusion states. Strong local audio-visual agreement provides useful correspondence evidence, yet a fused state also needs sufficient cross-layer and cross-modal support before it can reliably guide later fusion. This paper studies this issue through propagation-aware representation readiness and formulates premature perceptual commitment as a readiness-deficiency problem, where local plausibility, propagation influence, and support insufficiency jointly appear at an intermediate stage. We propose the Delayed Perceptual Commitment Network (DPC-Net), an encoder-level framework that estimates an observable readiness-deficiency surrogate, localizes the intervention-sensitive bottleneck, and applies support-aware correction with cross-layer and cross-modal evidence. DPC-Net preserves task-specific heads, losses, decoding modules, and evaluation protocols, making it applicable to different audio-visual tasks through encoder-side intervention. Experiments on audio-visual speech separation, audio-visual event localization, and audio-visual speech recognition show consistent improvements across reconstruction, localization, and recognition regimes. Further analyses on component contribution, selection criteria, counterfactual intervention, and readiness trajectories support the effectiveness of readiness-guided bottleneck correction.
\end{abstract}


\section{Introduction}
\IEEEPARstart{A}{udio-visual} learning builds task-relevant representations by integrating acoustic and visual cues that provide complementary evidence in temporal dynamics, articulatory structure, semantic correspondence, and robustness under adverse conditions~\cite{ref1, ref2, ref3, ref4, ref5}. Audio-visual speech separation, audio-visual event localization, and audio-visual speech recognition differ in output form, yet they commonly rely on encoder-based audio-visual representation learning. In architectures with explicit stage-wise fusion, cross-modal interaction is performed across multiple intermediate layers, and each fused state becomes part of the evidence basis for later integration~\cite{ref6, ref7, ref8}. This propagation process makes representation readiness an important issue: a fused state should provide useful local correspondence while also carrying sufficient support for subsequent fusion.

Existing encoder-based audio-visual systems have improved intermediate fusion through cross-attention, gating, feature modulation, reliability-aware weighting, and adapter-style interaction between audio and visual streams~\cite{ref10, ref11, ref12, ref13, ref14}. These mechanisms enhance local cross-modal compatibility and adapt interaction strength under varying input conditions. Their main decision, however, is usually made at the current fusion step. A locally well-matched fused state can gain strong influence over deeper layers before complementary, temporally delayed, or reliability-dependent evidence has been sufficiently consolidated~\cite{ref15, ref16, ref17, ref18}. In this case, easy-to-match cues may dominate later propagation, while weakly matched but task-relevant cues receive limited support during subsequent representation formation.

\begin{figure*}[t]
    \centering
    \includegraphics[width=0.95\textwidth]{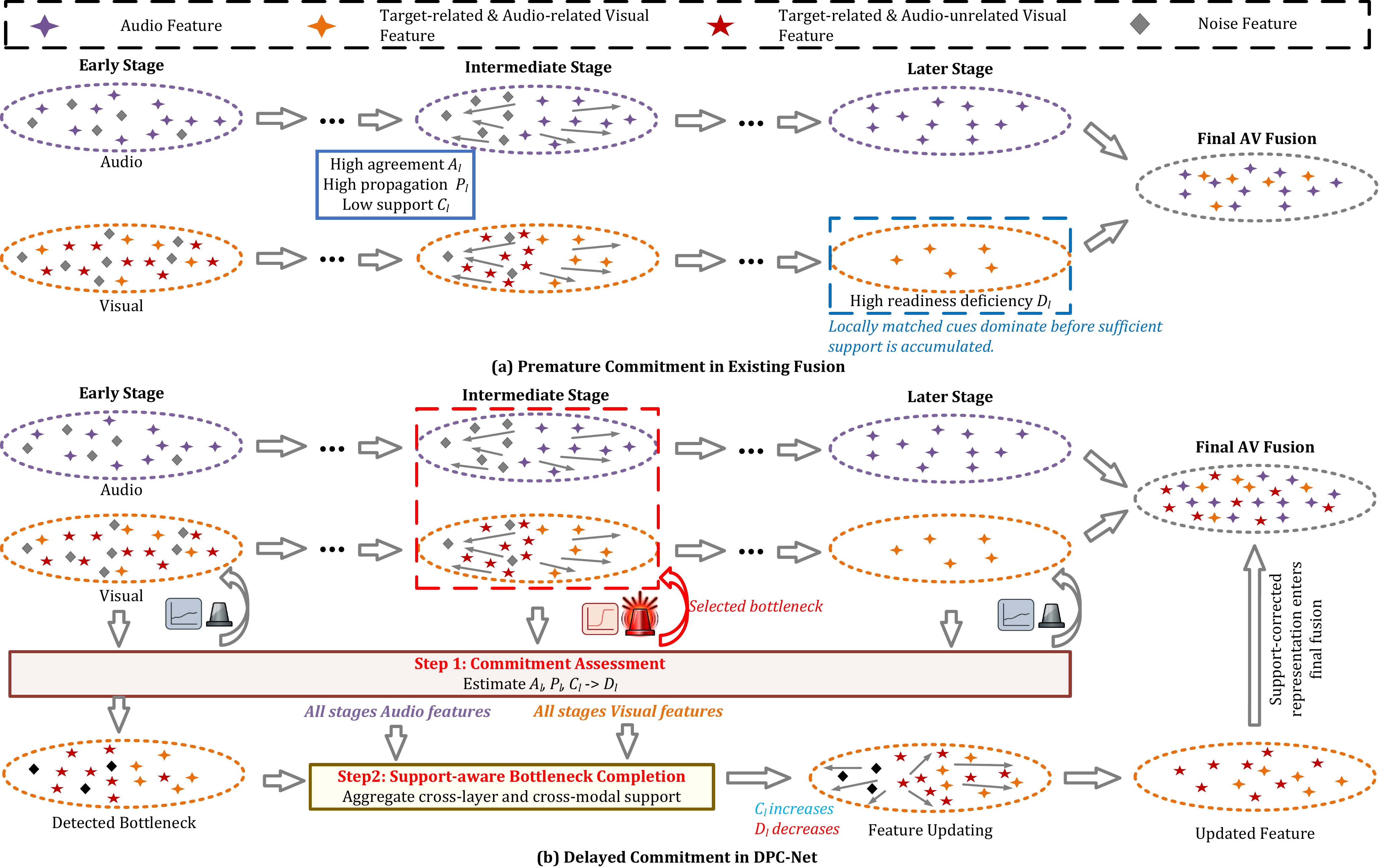}
    \caption{
    Conceptual illustration of premature commitment and delayed commitment.
    \textbf{(a)} In existing stage-wise fusion, a locally plausible fused state may gain propagation influence before sufficient support coverage is accumulated, producing high readiness deficiency.
    \textbf{(b)} DPC-Net estimates an observable readiness-deficiency surrogate, localizes the intervention-sensitive bottleneck, and applies support-aware correction to increase support coverage and reduce deficiency before subsequent fusion.
    }
    \label{fig:motivation_dpc}
\end{figure*}

We describe this failure mode as \emph{premature perceptual commitment}. As illustrated in Fig.~\ref{fig:motivation_dpc}, it occurs when three conditions coexist at an intermediate fusion stage: strong current audio-visual agreement, sufficient propagation influence over later fusion, and insufficient support coverage from cross-layer, cross-modal, or reliability-dependent evidence. Such a state is locally plausible, but its support is incomplete for later propagation. Human audiovisual perception offers a useful computational analogy for this distinction, since multisensory integration involves local correspondence, cue reliability, common-cause plausibility, and uncertainty under partial correspondence or sensory degradation~\cite{ref28, ref29, ref30, ref31, ref32, ref33, ref34}. We use perceptual commitment as a representation-level abstraction that separates local correspondence from stable integration readiness.

Motivated by this perspective, we propose the \emph{Delayed Perceptual Commitment Network} (DPC-Net), an encoder-level framework for readiness control in stage-wise audio-visual learning. DPC-Net estimates an observable readiness-deficiency surrogate from three cues: current audio-visual agreement, downstream anchoring, and support coverage. The stage with the largest estimated deficiency is treated as the intervention-sensitive bottleneck. DPC-Net then aggregates cross-layer and cross-modal support evidence and applies a gated residual correction to the selected bottleneck. This design focuses correction on the vulnerable intermediate state while preserving task-specific heads, losses, decoding modules, and evaluation protocols.

The proposed formulation leads to a compact mechanism-oriented evaluation. A useful readiness surrogate should select a stage that is more sensitive to perturbation, more responsive to support-aware completion, and distinguishable from stages selected by agreement-only, support-deficit-only, attention-response, fixed-depth, or random criteria. We evaluate this principle across three audio-visual output regimes: signal-level reconstruction with audio-visual speech separation, segment-level localization with audio-visual event localization, and sequence-level recognition with audio-visual speech recognition. These tasks provide complementary tests of whether encoder-side readiness control benefits stage-wise audio-visual representation learning across different supervision and prediction structures.

The main contributions of this work are summarized as follows:
\begin{itemize}
    \item We define \emph{propagation-aware representation readiness} as an intervention-oriented property for stage-wise audio-visual learning, where a locally matched fused state should also be sufficiently supported before guiding later fusion.

    \item We formulate \emph{premature perceptual commitment} through readiness deficiency, which jointly characterizes local plausibility, propagation influence, and support insufficiency at an intermediate fusion stage.

    \item We propose DPC-Net, an encoder-level readiness-control framework that estimates an observable surrogate \(\widehat{D}_l\), localizes the intervention-sensitive bottleneck, and performs support-aware correction with cross-layer and cross-modal evidence.

    \item We validate the proposed framework across reconstruction, localization, and recognition regimes, with mechanism-oriented analyses covering component contribution, selection criteria, counterfactual intervention, and readiness trajectories.
\end{itemize}

\section{Related Work}

\subsection{Stage-wise Audio-Visual Representation Learning}

Audio-visual learning exploits complementary acoustic and visual evidence for event understanding, speech separation, speech recognition, and robust multimodal perception~\cite{ref1, ref6, ref38, ref41, ref40}. Early systems commonly integrated modalities through feature concatenation, audio-visual conditioning, or prediction-level aggregation. Recent encoder-based architectures increasingly perform fusion across multiple intermediate stages, allowing audio and visual streams to interact during representation formation~\cite{ref6, ref38, ref39, ref42}. This design has been widely adopted in audio-visual event localization, speech separation, and speech recognition, where visual cues support event reasoning, target speaker extraction, and recognition under noisy acoustic conditions~\cite{ref4, ref7, ref22, ref25, ref41, ref43, ref44}.

Most stage-wise fusion methods focus on improving cross-modal exchange at each layer, such as strengthening correspondence, modeling temporal interaction, or refining modality-specific representations. Repeated fusion also creates a propagation process: once a fused state is formed, it can shape the evidence available to deeper layers. This work focuses on the readiness of such intermediate fused states, asking whether a locally plausible representation has accumulated enough support to guide subsequent fusion.

\subsection{Selective and Reliability-aware Fusion}

Selective and reliability-aware fusion methods regulate modality contribution, feature exchange, or interaction strength according to the current input condition~\cite{ref45, ref46, ref47, ref48, ref13, ref42}. Representative designs include reliability-aware weighting, modality-conditioned enhancement, trainable cross-modal adapters, prompt-based modulation, and mixture-of-experts routing~\cite{ref45, ref46, ref48, ref47, ref13, ref42}. These approaches improve robustness when one modality becomes noisy, incomplete, or less informative than the other.

The main decision in these methods is usually local to the current fusion step. Reliability, compatibility, confidence, or routing scores determine how the model modulates the present interaction~\cite{ref45, ref46, ref47, ref13}. Representation-readiness control addresses a complementary question: after a fused state is produced, whether it has enough cross-layer and cross-modal support to influence later propagation. This distinction is central to DPC-Net, which estimates the readiness deficiency of intermediate states and corrects the selected bottleneck before subsequent fusion proceeds.

\subsection{Intermediate Representation Diagnosis and Intervention}

Intermediate representation analysis helps reveal how deep models organize evidence across layers. Layer-wise probing, perturbation analysis, sensitivity measurement, and intervention-based diagnosis can indicate whether a hidden state carries useful information, how vulnerable it is to corruption, and how modifying it affects downstream prediction~\cite{ref15, ref16, ref17, ref19}. These analyses are relevant to stage-wise audio-visual learning because intermediate fused states participate in later representation formation.

Existing audio-visual fusion methods usually select interactions according to current compatibility, reliability, or attention responses. A readiness-oriented framework instead treats stage selection as an intervention problem. The key is to identify the intermediate state that is both influential for later propagation and insufficiently supported by available evidence. This motivates the bottleneck-localization view adopted in DPC-Net.

\subsection{Perceptual Commitment and Representation Readiness}

Human audiovisual perception provides a useful computational reference for separating local correspondence from stable commitment~\cite{ref28, ref30, ref31}. Studies of causal inference and reliability-aware integration suggest that multisensory perception depends on cross-modal agreement, common-cause plausibility, cue reliability, and uncertainty under partial correspondence or sensory degradation~\cite{ref30, ref31, ref49, ref37, ref50, ref34}. Related findings also indicate that integrated and segregated interpretations can remain in competition before a stable percept is formed~\cite{ref33, ref37, ref36}.

We adopt perceptual commitment as a representation-learning abstraction. For stage-wise audio-visual encoders, local agreement provides correspondence evidence, while propagation readiness further depends on support coverage from cross-layer, cross-modal, and reliability-dependent evidence. The delayed commitment formulation follows this view by treating premature commitment as a readiness failure in which a locally plausible intermediate state gains propagation influence before sufficient support coverage is accumulated.

\section{Propagation-aware Readiness Principle}
\label{sec:readiness_formulation}

This section develops the readiness-control principle for stage-wise audio-visual fusion. The central object is the readiness deficiency \(D_l\), which characterizes an intermediate fused state that is locally plausible, influential for later propagation, and insufficiently supported. The formulation first describes stage-wise fusion as evidence propagation, then defines readiness deficiency, introduces an observable surrogate for bottleneck localization, and finally derives intervention-oriented tests used in the mechanism analysis.

\subsection{Evidence Propagation in Stage-wise Fusion}
\label{sec:evidence_propagation}

Consider an encoder with \(L\) explicit audio-visual fusion stages. Let \(a_l\) and \(v_l\) denote the audio and visual states at stage \(l\). The fused state is written as
\begin{equation}
z_l = F_l(a_l,v_l),
\label{eq:stage_fusion}
\end{equation}
where \(F_l(\cdot)\) denotes the stage-wise fusion operator. Since intermediate states are propagated through the encoder, the next fusion stage can be abstractly expressed as
\begin{equation}
z_{l+1}=T_{l+1}(z_l,a_{l+1},v_{l+1}),
\label{eq:stage_propagation}
\end{equation}
where \(T_{l+1}(\cdot)\) denotes the subsequent propagation and fusion operator. Thus, \(z_l\) acts as an evidence state for deeper audio-visual integration.

This propagation view motivates three quantities for evaluating whether an intermediate fused state is ready to guide later fusion. The first quantity is current audio-visual agreement:
\begin{equation}
A_l = \mathcal{A}(a_l,v_l),
\label{eq:agreement}
\end{equation}
where \(\mathcal{A}(\cdot)\) measures local correspondence between the two modalities. The second quantity is propagation influence:
\begin{equation}
P_l = \mathcal{P}(z_l,z_{>l}),
\label{eq:propagation_influence}
\end{equation}
where \(z_{>l}\) denotes the later fusion trajectory affected by \(z_l\). A high \(P_l\) indicates that the current fused state can strongly shape subsequent representation formation.

The third quantity is support coverage. Let
\begin{equation}
\mathcal{E}_a=\{a_j\}_{j=1}^{L},
\qquad
\mathcal{E}_v=\{v_j\}_{j=1}^{L}
\label{eq:all_stage_evidence}
\end{equation}
denote the all-stage audio and visual evidence. The supportive evidence for stage \(l\) is obtained as
\begin{equation}
q_l = H_l(\mathcal{E}_a,\mathcal{E}_v,z_l),
\label{eq:support_evidence}
\end{equation}
where \(H_l(\cdot)\) aggregates cross-layer and cross-modal information relevant to \(z_l\). Support coverage is defined as
\begin{equation}
C_l = \mathcal{C}(z_l,q_l),
\label{eq:support_coverage}
\end{equation}
where \(\mathcal{C}(\cdot)\) measures how sufficiently the current fused state is covered by supportive evidence.

\begin{definition}[Propagation-aware representation readiness]
\label{def:representation_readiness}
A fused state \(z_l\) is propagation-ready when its local plausibility, propagation influence, and support coverage are jointly sufficient for subsequent fusion. It is readiness-deficient when the state is locally plausible and influential for later propagation, while its support coverage remains insufficient. In this definition, \(A_l\) characterizes local plausibility, \(P_l\) characterizes propagation influence, and \(C_l\) characterizes support sufficiency.
\end{definition}

\subsection{Readiness Deficiency}
\label{sec:readiness_deficiency}

Based on Definition~\ref{def:representation_readiness}, we define the propagation-aware readiness deficiency as
\begin{equation}
D_l =
\underbrace{\psi(A_l-\tau_A)}_{\text{local plausibility}}
\underbrace{\psi(P_l-\tau_P)}_{\text{propagation influence}}
\underbrace{\psi(\tau_C-C_l)}_{\text{support insufficiency}},
\label{eq:readiness_deficiency}
\end{equation}
where \(\tau_A\), \(\tau_P\), and \(\tau_C\) are thresholds for agreement, propagation influence, and support coverage, respectively. The function \(\psi(\cdot)\) is a non-negative activation, such as a positive-part function or its smooth approximation. Equation~\eqref{eq:readiness_deficiency} becomes large when strong local plausibility, sufficient propagation influence, and support insufficiency appear together.

This formulation separates readiness deficiency from simpler criteria. A high-agreement stage can already be well supported, a low-support stage can have limited influence on later fusion, and a high-propagation stage can already possess sufficient support. The readiness bottleneck is therefore the stage where local plausibility, propagation influence, and support insufficiency jointly become most severe. Under the common positive-part condition, where \(\psi(x)=0\) for \(x\leq 0\) and \(\psi(x)>0\) for \(x>0\), \(D_l>0\) holds exactly when
\begin{equation}
A_l>\tau_A,\qquad
P_l>\tau_P,\qquad
C_l<\tau_C.
\label{eq:joint_activation}
\end{equation}
This joint activation condition explains why readiness control differs from selecting the highest-agreement stage, the lowest-support stage, or a fixed intermediate layer.

\subsection{Observable Surrogate and Bottleneck Localization}
\label{sec:bottleneck_localization}

The analytical readiness bottleneck is defined as
\begin{equation}
l^\star=\arg\max_{l} D_l .
\label{eq:analytical_bottleneck}
\end{equation}
In implementation, DPC-Net estimates feature-level cues and uses them to approximate the intervention-relevant ordering of fusion stages. The agreement cue is computed as
\begin{equation}
\widehat{A}_l =
\operatorname{sim}
\left(
\phi_a(a_l),
\phi_v(v_l)
\right).
\label{eq:agreement_proxy}
\end{equation}
The propagation cue is implemented as a downstream-anchoring proxy:
\begin{equation}
\widehat{P}_l =
\operatorname{sim}
\left(
\phi_p(z_l),
\phi_p(f_f)
\right),
\label{eq:propagation_proxy}
\end{equation}
where \(f_f\) denotes the fused summary delivered by the encoder to the subsequent task module. This cue measures how strongly the stage representation is anchored to the downstream fused representation, making it useful for intervention ranking. The support-coverage cue is computed as
\begin{equation}
\widehat{C}_l =
\operatorname{sim}
\left(
\phi_c(z_l),
\phi_c(q_l)
\right).
\label{eq:support_proxy}
\end{equation}
The observable readiness-deficiency surrogate is then
\begin{equation}
\widehat{D}_l =
\psi(\widehat{A}_l-\tau_A)
\psi(\widehat{P}_l-\tau_P)
\psi(\tau_C-\widehat{C}_l).
\label{eq:observable_deficiency}
\end{equation}
The role of \(\widehat{D}_l\) is intervention-oriented. It ranks stages according to their expected vulnerability to perturbation and responsiveness to support-aware correction. Since DPC-Net performs stage intervention, the key requirement for \(\widehat{D}_l\) is ordering consistency.

\begin{proposition}[Ordering-preserving bottleneck localization]
\label{prop:ordering_localization}
Assume that the analytical bottleneck stage is unique, \(l^\star = \arg\max_l D_l\). Let \(\varphi(\cdot)\) be a strictly increasing calibration function, and define
\begin{equation}
m_\varphi =
\varphi(D_{l^\star})
-
\max_{j\ne l^\star}\varphi(D_j).
\label{eq:calibrated_margin}
\end{equation}
Suppose that \(\widehat{D}_l=\varphi(D_l)+e_l\) and \(|e_l|\leq \epsilon\) for all \(l\). If \(m_\varphi>2\epsilon\), then
\begin{equation}
\widehat{l}
=
\arg\max_l \widehat{D}_l
=
l^\star .
\label{eq:surrogate_localization}
\end{equation}
\end{proposition}

\begin{proof}
For \(l^\star\), we have \(\widehat{D}_{l^\star}\geq \varphi(D_{l^\star})-\epsilon\). For any \(j\ne l^\star\), \(\widehat{D}_j\leq \max_{k\ne l^\star}\varphi(D_k)+\epsilon\). The margin condition \(m_\varphi>2\epsilon\) gives \(\widehat{D}_{l^\star}>\widehat{D}_j\) for all \(j\ne l^\star\), which proves Eq.~\eqref{eq:surrogate_localization}.
\end{proof}

Proposition~\ref{prop:ordering_localization} shows that the surrogate only needs to preserve the intervention-relevant ordering of stages with sufficient margin. This property matches the goal of readiness control: identifying the stage most likely to be harmful when perturbed and beneficial when corrected.

\subsection{Support-aware Bottleneck Correction}
\label{sec:support_correction}

After localizing \(l^\star\), DPC-Net improves the support coverage of the selected fused state. The support evidence for the selected stage is
\begin{equation}
q_{l^\star}
=
H_{l^\star}(\mathcal{E}_a,\mathcal{E}_v,z_{l^\star}).
\label{eq:selected_support}
\end{equation}
A correction direction is constructed from the selected state and its support evidence:
\begin{equation}
d_{l^\star}
=
\mathcal{B}(z_{l^\star},q_{l^\star}),
\label{eq:completion_direction}
\end{equation}
where \(\mathcal{B}(\cdot)\) denotes the completion operator. The corrected state is written as
\begin{equation}
\widetilde{z}_{l^\star}(\lambda)
=
z_{l^\star}
+
\lambda d_{l^\star},
\qquad
\lambda\in[0,1].
\label{eq:corrected_state}
\end{equation}
This equation abstracts the gated residual update used in the network implementation, with the purpose of increasing support coverage for the selected bottleneck.

\begin{proposition}[Local support correction reduces readiness deficiency]
\label{prop:support_correction}
Assume that \(C(\cdot)\) is continuously differentiable near \(z_{l^\star}\). If
\begin{equation}
\nabla C(z_{l^\star})^\top d_{l^\star}>0,
\label{eq:support_direction}
\end{equation}
then there exists \(\lambda_0\in(0,1]\) such that, for every \(\lambda\in(0,\lambda_0]\),
\begin{equation}
\widetilde{C}_{l^\star}(\lambda)>C_{l^\star}.
\label{eq:support_increase}
\end{equation}
Furthermore, if the active agreement-propagation factor satisfies
\begin{equation}
\begin{aligned}
&\psi(\widetilde{A}_{l^\star}-\tau_A)
\psi(\widetilde{P}_{l^\star}-\tau_P)
\\
&\quad \leq
\psi(A_{l^\star}-\tau_A)
\psi(P_{l^\star}-\tau_P),
\end{aligned}
\label{eq:active_factor}
\end{equation}
then \(\widetilde{D}_{l^\star}(\lambda)<D_{l^\star}\).
\end{proposition}

\begin{proof}
Let \(\chi(\lambda)=C(z_{l^\star}+\lambda d_{l^\star})\). A first-order expansion gives
\begin{equation}
\chi(\lambda)
=
C(z_{l^\star})
+
\lambda
\nabla C(z_{l^\star})^\top d_{l^\star}
+
o(\lambda).
\label{eq:first_order_support}
\end{equation}
Eq.~\eqref{eq:support_direction} ensures that \(\chi(\lambda)>C(z_{l^\star})\) for sufficiently small positive \(\lambda\). Since \(\psi(\tau_C-C)\) decreases as \(C\) increases, the support-insufficiency factor decreases. Combining this with Eq.~\eqref{eq:active_factor} gives the reduction of readiness deficiency.
\end{proof}

Proposition~\ref{prop:support_correction} gives a local sufficient condition for bottleneck completion. It motivates the readiness-trajectory analysis, which checks whether support-aware correction increases \(\widehat{C}_{l^\star}\) while reducing \(\widehat{D}_{l^\star}\) at the selected stage.

\begin{table}[t]
\centering
\caption{
Connection between readiness formulation, implementation, and validation.
}
\label{tab:readiness_mapping}
\setlength{\tabcolsep}{3.5pt}
\renewcommand{\arraystretch}{1.08}
\resizebox{\linewidth}{!}{%
\begin{tabular}{lll}
\hline
Concept & Implementation cue & Validation focus \\
\hline
Local plausibility & \(\widehat{A}_l\): audio-visual agreement & Agreement-only comparison \\
Propagation relevance & \(\widehat{P}_l\): downstream anchoring & Intervention sensitivity \\
Support sufficiency & \(\widehat{C}_l\): support coverage & Readiness trajectory \\
Readiness deficiency & \(\widehat{D}_l\): joint surrogate & Bottleneck selection and recovery \\
\hline
\end{tabular}}
\end{table}

\subsection{Observable Tests for Readiness Control}
\label{sec:observable_implications}

The formulation leads to observable tests for validating \(\widehat{D}_l\) as an intervention-oriented surrogate. Let \(U(\cdot)\) denote a task utility, where larger values indicate better performance. For a controlled perturbation \(\Pi_l(\cdot)\), the utility drop is
\begin{equation}
\Delta^{\mathrm{drop}}_l
=
U(z)-U(\Pi_l(z)).
\label{eq:utility_drop}
\end{equation}
For completion \(\Gamma_l(\cdot)\), the utility recovery is
\begin{equation}
\Delta^{\mathrm{rec}}_l
=
U(\Gamma_l(\Pi_l(z)))-U(\Pi_l(z)).
\label{eq:utility_recovery}
\end{equation}
A valid readiness bottleneck should satisfy
\begin{equation}
\Delta^{\mathrm{drop}}_{l^\star}
>
\Delta^{\mathrm{drop}}_j,
\qquad
\Delta^{\mathrm{rec}}_{l^\star}
>
\Delta^{\mathrm{rec}}_j,
\qquad
j\in\mathcal{S}_{\mathrm{ctrl}},
\label{eq:observable_predictions}
\end{equation}
where \(\mathcal{S}_{\mathrm{ctrl}}\) denotes control stages selected by neighboring, random, fixed-depth, or alternative score-based strategies. In addition, \(\widehat{D}_l\) should provide a stronger intervention ranking than agreement-only, propagation-only, support-deficit-only, attention-response, fixed-depth, and random selection. Under degraded or imbalanced conditions, \(\widehat{C}_l\) should decrease at vulnerable stages while \(\widehat{D}_l\) becomes more concentrated. These tests guide the mechanism-oriented analyses in Section~\ref{sec:mechanism_analysis}. Table~\ref{tab:readiness_mapping} summarizes how the readiness factors are instantiated and validated. This mapping clarifies that \(\widehat{D}_l\) is evaluated as an intervention-oriented surrogate: each cue contributes to stage ranking, and the final score is tested through selection behavior, perturbation sensitivity, completion recovery, and readiness trajectories.

\section{Framework Implementation}
\label{sec:framework_implementation}

DPC-Net implements readiness control inside an encoder with explicit stage-wise audio-visual fusion. As shown in Fig.~\ref{fig:framework}, the framework contains two coupled components: commitment assessment and support-aware bottleneck completion. Commitment assessment estimates the readiness-deficiency surrogate of each fusion stage and selects the intervention-sensitive bottleneck. Support-aware bottleneck completion then aggregates cross-layer and cross-modal evidence to correct the selected representation before it is delivered to later fusion or the encoder output interface. Since the intervention is performed at the encoder level, task-specific heads, losses, decoding modules, and evaluation protocols are preserved.

\begin{figure*}[t]
    \centering
    \includegraphics[width=0.98\textwidth]{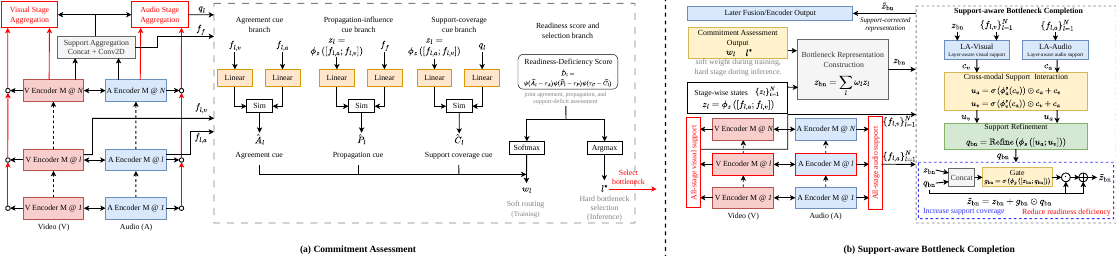}
    \caption{
    Implementation of DPC-Net.
    \textbf{(a)} Commitment assessment estimates agreement, downstream-anchoring, and support-coverage cues, and combines them into the readiness-deficiency surrogate for soft routing during training and hard bottleneck selection during inference.
    \textbf{(b)} Support-aware bottleneck completion aggregates all-stage audio-visual support and produces a support-corrected representation for later fusion or encoder output.
    }
    \label{fig:framework}
\end{figure*}

\subsection{Commitment Assessment}

Given \(N\) stage-wise audio and visual features \(\{f_{l,a}\}_{l=1}^{N}\) and \(\{f_{l,v}\}_{l=1}^{N}\), commitment assessment evaluates the readiness state of each intermediate fusion stage. The fused state at stage \(l\) is constructed as
\begin{equation}
z_l = \phi_z([f_{l,a};f_{l,v}]),
\label{eq:impl_stage_state}
\end{equation}
where \([\cdot;\cdot]\) denotes feature concatenation and \(\phi_z(\cdot)\) is a lightweight projection. The downstream fused summary is computed from the final-stage audio and visual features:
\begin{equation}
f_f = \phi_f([f_{N,a};f_{N,v}]).
\label{eq:impl_fused_summary}
\end{equation}
This summary provides the downstream reference for estimating how strongly an intermediate state is anchored to the representation delivered to the task module.

For each stage, DPC-Net estimates three cues corresponding to the readiness factors in Section~\ref{sec:readiness_formulation}. The agreement cue is
\begin{equation}
\widehat{A}_l =
\operatorname{sim}
\left(
\phi_a(f_{l,a}),
\phi_v(f_{l,v})
\right),
\label{eq:impl_agreement}
\end{equation}
where \(\operatorname{sim}(\cdot,\cdot)\) denotes cosine similarity after projection. The downstream-anchoring cue is
\begin{equation}
\widehat{P}_l =
\operatorname{sim}
\left(
\phi_p(z_l),
\phi_p(f_f)
\right).
\label{eq:impl_propagation}
\end{equation}
A higher \(\widehat{P}_l\) indicates that the stage representation is more strongly aligned with the downstream fused summary and is therefore more relevant for intervention ranking.

The support-coverage cue is obtained by first aggregating all-stage support evidence:
\begin{equation}
q_l =
H_l
\left(
z_l,
\{f_{j,a}\}_{j=1}^{N},
\{f_{j,v}\}_{j=1}^{N}
\right),
\label{eq:impl_support_aggregation}
\end{equation}
where \(H_l(\cdot)\) denotes the support aggregation function. The support-coverage cue is then computed as
\begin{equation}
\widehat{C}_l =
\operatorname{sim}
\left(
\phi_c(z_l),
\phi_c(q_l)
\right).
\label{eq:impl_support_coverage}
\end{equation}
The three cues are combined into the observable readiness-deficiency surrogate:
\begin{equation}
\widehat{D}_l =
\psi(\widehat{A}_l-\tau_A)
\psi(\widehat{P}_l-\tau_P)
\psi(\tau_C-\widehat{C}_l),
\label{eq:impl_readiness_score}
\end{equation}
where \(\tau_A\), \(\tau_P\), and \(\tau_C\) are thresholds in the normalized score space. We implement \(\psi(\cdot)\) with a smooth positive-part approximation:
\begin{equation}
\psi(x)=\frac{1}{\beta_s}\log(1+\exp(\beta_s x)),
\label{eq:smooth_positive}
\end{equation}
where \(\beta_s\) controls the sharpness. Thus, \(\widehat{D}_l\) becomes large when agreement, downstream anchoring, and support insufficiency appear together.

During training, the stage scores are converted into differentiable routing weights:
\begin{equation}
\omega_l =
\frac{\exp(\widehat{D}_l/\tau_s)}
{\sum_{j=1}^{N}\exp(\widehat{D}_j/\tau_s)},
\label{eq:soft_routing}
\end{equation}
where \(\tau_s\) is the routing temperature. During inference, DPC-Net performs hard bottleneck selection:
\begin{equation}
l^\star = \arg\max_l \widehat{D}_l .
\label{eq:hard_selection}
\end{equation}
The selected stage is then passed to support-aware bottleneck completion.

\subsection{Support-aware Bottleneck Completion}

Support-aware bottleneck completion corrects the representation selected by commitment assessment. During training, the bottleneck representation is constructed by soft routing:
\begin{equation}
z_{\mathrm{bn}}
=
\sum_{l=1}^{N}
\omega_l z_l .
\label{eq:bottleneck_train}
\end{equation}
During inference, \(\omega_l\) becomes a one-hot selection induced by \(l^\star\), and \(z_{\mathrm{bn}}=z_{l^\star}\). This design provides differentiable training and stage-specific intervention at test time.

Given \(z_{\mathrm{bn}}\), DPC-Net aggregates support from all visual and audio stages. The layer-aware visual support is
\begin{equation}
c_v =
\operatorname{LA\mbox{-}Visual}
\left(
z_{\mathrm{bn}},
\{f_{l,v}\}_{l=1}^{N}
\right),
\label{eq:visual_support}
\end{equation}
and the layer-aware audio support is
\begin{equation}
c_a =
\operatorname{LA\mbox{-}Audio}
\left(
z_{\mathrm{bn}},
\{f_{l,a}\}_{l=1}^{N}
\right).
\label{eq:audio_support}
\end{equation}
These operations collect cross-layer evidence conditioned on the selected bottleneck representation.

The two support representations are refined through cross-modal support interaction:
\begin{equation}
u_a =
\sigma(\phi_a^s(c_v))\odot c_a + c_v,
\label{eq:audio_support_interaction}
\end{equation}
\begin{equation}
u_v =
\sigma(\phi_v^s(c_a))\odot c_v + c_a,
\label{eq:visual_support_interaction}
\end{equation}
where \(\sigma(\cdot)\) denotes the sigmoid function and \(\odot\) denotes element-wise multiplication. The refined support evidence is obtained as
\begin{equation}
q_{\mathrm{bn}}
=
\operatorname{Refine}
\left(
\phi_s([u_a;u_v])
\right).
\label{eq:refined_support}
\end{equation}
Finally, DPC-Net applies a gated residual correction:
\begin{equation}
g_{\mathrm{bn}}
=
\sigma
\left(
\phi_g([z_{\mathrm{bn}};q_{\mathrm{bn}}])
\right),
\label{eq:residual_gate}
\end{equation}
\begin{equation}
\widetilde{z}_{\mathrm{bn}}
=
z_{\mathrm{bn}}
+
g_{\mathrm{bn}}\odot q_{\mathrm{bn}}.
\label{eq:support_corrected_state}
\end{equation}
The gate controls how much supportive evidence is injected into the bottleneck representation. This update increases the support available to the selected stage and produces the support-corrected representation used by later fusion or the encoder output interface.

\subsection{Encoder-level Deployment}

DPC-Net is deployed as an encoder-level intervention. Commitment assessment produces \(\{\omega_l\}_{l=1}^{N}\) during training and \(l^\star\) during inference. Support-aware bottleneck completion then constructs \(z_{\mathrm{bn}}\), aggregates all-stage support, and outputs \(\widetilde{z}_{\mathrm{bn}}\). The corrected representation is delivered to the later fusion or encoder output module used by each task.

This deployment makes DPC-Net compatible with different audio-visual tasks. In AVSS, AVEL, and AVSR, the framework is inserted into the audio-visual encoder, while the task-specific prediction head, loss function, and decoding process remain unchanged. The same readiness surrogate also provides an explicit intervention target, enabling the selection-criterion, counterfactual, and trajectory analyses in Section~\ref{sec:mechanism_analysis}.

\section{Experimental Setup}
\label{sec:experimental_setup}

We evaluate DPC-Net on three audio-visual output regimes: signal-level reconstruction with audio-visual speech separation (AVSS), segment-level localization with audio-visual event localization (AVEL), and sequence-level recognition with audio-visual speech recognition (AVSR). These tasks differ in supervision form, output structure, and evaluation metric, while sharing a stage-wise audio-visual encoder in which intermediate fused states are propagated to deeper representations. As shown in Fig.~\ref{fig:task_deployment}, DPC-Net is consistently inserted into the audio-visual encoder, and the task-specific heads, losses, decoding modules, and evaluation protocols are preserved.

\begin{figure*}[t]
    \centering
    \includegraphics[width=0.98\textwidth]{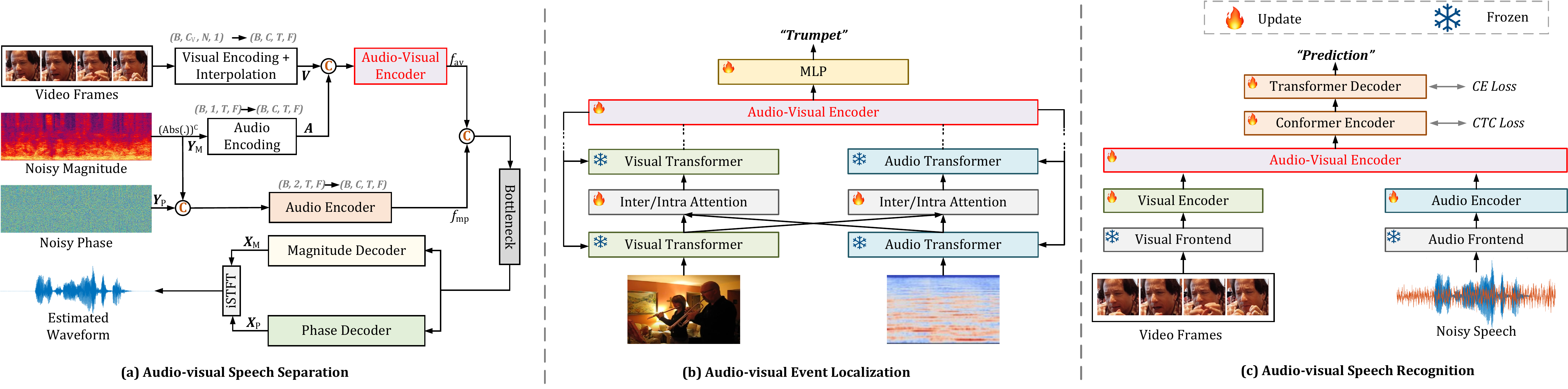}
    \caption{
    Task-level deployment of DPC-Net across AVSS, AVEL, and AVSR.
    DPC-Net is inserted into the audio-visual encoder in all three tasks, while task-specific heads, decoders, losses, and prediction modules are preserved.
    }
    \label{fig:task_deployment}
\end{figure*}

\subsection{Tasks, Datasets, and Metrics}

\textbf{AVSS.}
For audio-visual speech separation, we use a dual encoder-decoder backbone that encodes noisy speech and synchronized video frames before audio-visual fusion. DPC-Net is introduced into the audio-visual encoder, while the phase-magnitude encoding, separation blocks, decoding path, and waveform reconstruction module follow the original backbone~\cite{ref22, ref25, ref41}. We evaluate AVSS on LRS2~\cite{ref4}, LRS3~\cite{ref43}, and VoxCeleb2~\cite{ref44}, following commonly used protocols in recent AVSS studies~\cite{ref10, ref45, ref46}. SI-SNRi, SDRi~\cite{ref47}, and PESQ~\cite{ref48} are reported. To evaluate robustness under unreliable visual evidence, we consider random patch occlusion and noise+blur corruption~\cite{ref22, ref49}.

\textbf{AVEL.}
For audio-visual event localization, we use a paired audio-visual transformer backbone with stage-wise inter-modal and intra-modal interaction. DPC-Net is inserted into the audio-visual encoder before the segment-level prediction head. We evaluate AVEL on the AVE benchmark~\cite{ref41} in the fully supervised setting and report segment-level classification accuracy. The experiments cover shared-backbone and separate-backbone configurations, including ViT, Swin-V2, and HTS-AT settings~\cite{ref50}.

\textbf{AVSR.}
For audio-visual speech recognition, the backbone uses modality-specific front-ends and encoders before audio-visual fusion~\cite{ref51}. DPC-Net is inserted into the audio-visual encoder before the downstream Conformer encoder, CTC branch, and Transformer decoder~\cite{ref54}. We evaluate AVSR on LRS2~\cite{ref4} and LRS3~\cite{ref43}, and report word error rate (WER) under clean and noisy acoustic conditions. We also evaluate joint audio-visual corruption on LRS2 by combining low-SNR audio with patch occlusion or noise+blur visual degradation.

\subsection{Implementation and Comparison Protocol}

All models are optimized with Adam~\cite{ref56}. We follow the task-specific training objectives, decoding procedures, and learning-rate schedules used by the corresponding backbones. Training is conducted on NVIDIA H100 GPUs, and early stopping is applied according to the validation metric of each task.

For benchmark comparison, we follow the standard dataset splits, evaluation metrics, and reporting conventions of each task. Baseline results are taken from the original papers when their reported settings match the benchmark protocol. For the strongest or most directly comparable baselines marked with \(\dagger\), we re-implement and re-train them under the same data split, preprocessing pipeline, optimizer setting, training schedule, and evaluation protocol as DPC-Net. This protocol distinguishes controlled comparisons from results cited under standard benchmark settings.

\subsection{Mechanism-oriented Evaluation}

Beyond benchmark performance, we evaluate whether the learned readiness-deficiency surrogate \(\widehat{D}_l\) behaves as an intervention-oriented bottleneck score. The analysis includes four aspects: component ablation, selection-criterion validation, counterfactual intervention, and readiness-trajectory analysis. Component ablation examines the roles of commitment assessment and support-aware bottleneck completion. Selection-criterion validation compares \(\widehat{D}_l\) with agreement-only, propagation-only, support-deficit-only, attention-response, fixed-depth, and random selection. Counterfactual intervention tests whether perturbing the selected stage causes a larger utility drop and whether completing it yields stronger recovery. Readiness-trajectory analysis examines how \(\widehat{A}_l\), \(\widehat{P}_l\), \(\widehat{C}_l\), and \(\widehat{D}_l\) change under degraded conditions and after support-aware correction.

For metrics with different directions, we follow the standard interpretation of each task: higher SI-SNRi, SDRi, PESQ, and accuracy indicate better performance, while lower WER indicates better recognition. Perturbation-induced degradation and completion-induced recovery are compared within each task because the metrics have different units and scales.

\section{Main Results}
\label{sec:main_evaluation}

We compare DPC-Net with representative audio-visual baselines across three output regimes: signal-level reconstruction, segment-level localization, and sequence-level recognition. In all settings, DPC-Net is inserted into the audio-visual encoder, while task-specific heads, losses, decoding modules, and evaluation protocols are preserved. The results therefore evaluate whether readiness-guided encoder intervention improves stage-wise audio-visual representation learning across different prediction structures.

\subsection{Signal-level Reconstruction: AVSS}
\label{sec:avss_results}

\begin{table*}[t]
\centering
\caption{
AVSS evaluation under standard visual conditions. Higher SI-SNRi, SDRi, and PESQ indicate better performance.
}
\label{tab:avss_standard}
\setlength{\tabcolsep}{3.2pt}
\renewcommand{\arraystretch}{1.05}
\begin{tabular}{lccccccccccccc}
\hline
\multirow{2}{*}{Method} &
\multicolumn{4}{c}{Efficiency} &
\multicolumn{3}{c}{LRS2} &
\multicolumn{3}{c}{LRS3} &
\multicolumn{3}{c}{VoxCeleb2} \\
\cline{2-14}
& Params & MACs & GPU & CPU
& SI-SNRi & SDRi & PESQ
& SI-SNRi & SDRi & PESQ
& SI-SNRi & SDRi & PESQ \\
& (M) & (G) & (ms) & (s)
& \(\uparrow\) & \(\uparrow\) & \(\uparrow\)
& \(\uparrow\) & \(\uparrow\) & \(\uparrow\)
& \(\uparrow\) & \(\uparrow\) & \(\uparrow\) \\
\hline
AV-ConvTasNet~\cite{ref62} & 16.5 & 23.8 & 118.77 & 1.22 & 12.5 & 12.8 & 2.69 & 11.2 & 11.7 & 2.58 & 9.2 & 9.8 & 2.17 \\
VisualVoice~\cite{ref34} & 77.8 & 9.7 & 231.65 & 3.04 & 11.5 & 11.8 & 2.78 & 9.9 & 10.3 & 2.13 & 9.3 & 10.2 & 2.45 \\
CaffNet-C~\cite{ref5} & -- & -- & -- & -- & -- & 10.0 & 1.15 & -- & 9.8 & -- & -- & 7.6 & -- \\
CTC-Net~\cite{ref50} & 7.0 & 167.1 & 162.45 & 1.69 & 14.3 & 14.6 & 3.08 & 17.4 & 17.5 & 3.24 & 11.9 & 13.1 & 3.00 \\
AVLiT-8~\cite{ref63} & 5.8 & 18.2 & 116.27 & 1.15 & 12.8 & 13.1 & 2.56 & 13.5 & 13.6 & 2.78 & 9.4 & 9.9 & 2.23 \\
RTFS-Net-12~\cite{ref48} & 0.7 & 56.4 & 144.61 & 1.52 & 14.9 & 15.1 & 3.07 & 17.5 & 17.6 & 3.25 & 12.4 & 13.6 & 3.00 \\
IIANet\(\dagger\)~\cite{ref11} & 3.1 & 18.6 & 238.94 & 1.46 & 16.2 & 16.4 & 3.26 & 18.5 & 18.7 & 3.31 & 13.8 & 14.5 & 3.15 \\
AV-CrossNet\(\dagger\)~\cite{ref49} & 11.1 & 29.6 & 392.0 & 2.46 & 16.5 & 16.9 & 3.34 & 18.6 & 18.8 & 3.44 & 14.2 & 14.6 & 3.23 \\
\textbf{DPC-Net (Ours)} & 7.1 & 15.7 & 133.12 & 1.24 & \textbf{16.8} & \textbf{17.3} & \textbf{3.52} & \textbf{18.9} & \textbf{19.2} & \textbf{3.69} & \textbf{14.7} & \textbf{14.9} & \textbf{3.39} \\
\hline
\end{tabular}
\end{table*}

Table~\ref{tab:avss_standard} compares DPC-Net with representative AVSS methods under standard visual conditions. DPC-Net achieves the best overall separation performance across LRS2, LRS3, and VoxCeleb2, while maintaining a moderate parameter and computation budget. The gains over the strongest controlled baselines are obtained with the same separation objective, decoding path, and waveform reconstruction module. This pattern indicates that readiness-guided correction improves the intermediate audio-visual representation formed by the encoder.

\begin{table*}[t]
\centering
\caption{
AVSS robustness evaluation under degraded visual input. Occlusion and noise+blur provide stress tests for unreliable visual evidence.
}
\label{tab:avss_robustness}
\setlength{\tabcolsep}{2.2pt}
\renewcommand{\arraystretch}{1.05}
\resizebox{\textwidth}{!}{%
\begin{tabular}{lcccccccccccccccccc}
\hline
\multirow{3}{*}{Method} &
\multicolumn{9}{c}{Occlusion} &
\multicolumn{9}{c}{Noise + Blur} \\
\cline{2-19}
& \multicolumn{3}{c}{LRS2}
& \multicolumn{3}{c}{LRS3}
& \multicolumn{3}{c}{VoxCeleb2}
& \multicolumn{3}{c}{LRS2}
& \multicolumn{3}{c}{LRS3}
& \multicolumn{3}{c}{VoxCeleb2} \\
\cline{2-19}
& SI-SNRi & SDRi & PESQ
& SI-SNRi & SDRi & PESQ
& SI-SNRi & SDRi & PESQ
& SI-SNRi & SDRi & PESQ
& SI-SNRi & SDRi & PESQ
& SI-SNRi & SDRi & PESQ \\
\hline
AV-ConvTasNet-LQ~\cite{ref64}
& 12.8 & 12.6 & 2.72
& 11.2 & 11.6 & 2.83
& 9.1 & 9.3 & 2.63
& 13.3 & 13.0 & 2.82
& 13.8 & 13.1 & 2.85
& 9.5 & 9.9 & 2.71 \\
MHSA-CRN~\cite{ref65}
& 12.5 & 12.9 & 2.88
& 11.8 & 12.1 & 2.99
& 10.5 & 10.6 & 2.74
& 13.8 & 13.4 & 2.79
& 13.1 & 13.0 & 2.77
& 9.2 & 9.8 & 2.66 \\
RAVSS~\cite{ref66}
& 13.1 & 13.9 & 3.00
& 14.3 & 14.5 & 3.03
& 11.7 & 11.8 & 2.93
& 13.8 & 14.1 & 3.01
& 14.0 & 14.3 & 3.04
& 12.0 & 12.2 & 2.99 \\
\textbf{DPC-Net (Ours)}
& \textbf{14.8} & \textbf{15.3} & \textbf{3.15}
& \textbf{16.4} & \textbf{17.2} & \textbf{3.18}
& \textbf{12.5} & \textbf{13.2} & \textbf{3.01}
& \textbf{14.9} & \textbf{15.9} & \textbf{3.17}
& \textbf{16.7} & \textbf{17.1} & \textbf{3.20}
& \textbf{13.2} & \textbf{13.7} & \textbf{3.09} \\
\hline
\end{tabular}}
\end{table*}

Table~\ref{tab:avss_robustness} evaluates robustness under visual degradation. DPC-Net consistently improves SI-SNRi, SDRi, and PESQ under occlusion and noise+blur across the evaluated datasets. These conditions weaken visual reliability and make local audio-visual agreement less stable as an indicator of later fusion support. The stronger results under degradation align with the readiness-control motivation: identifying and correcting the under-supported bottleneck helps maintain a useful representation for speech reconstruction.

\subsection{Segment-level Localization: AVEL}
\label{sec:avel_results}

\begin{table*}[t]
\centering
\caption{
AVEL evaluation on AVE in the fully supervised setting. Results are grouped by backbone configuration, and higher accuracy indicates better performance.
}
\label{tab:avel_results}
\setlength{\tabcolsep}{3.0pt}
\renewcommand{\arraystretch}{1.04}
\resizebox{\textwidth}{!}{%
\begin{tabular}{lccccccc}
\hline
Method & Visual Encoder & Audio Encoder & Visual Pretrain & Audio Pretrain & Trainable Params & Total Params & Acc. \(\uparrow\) \\
& & & & & (M) & (M) & (\%) \\
\hline
\multicolumn{8}{l}{\textit{CNN / conventional backbones}} \\
AVEL~\cite{ref67} & ResNet-152 & VGGish & ImageNet & AudioSet & 3.7 & 136.0 & 74.0 \\
AVSDN~\cite{ref68} & ResNet-152 & VGGish & ImageNet & AudioSet & 8.0 & 140.3 & 75.4 \\
CMRAN~\cite{ref69} & ResNet-152 & VGGish & ImageNet & AudioSet & 15.9 & 148.2 & 78.3 \\
MM-Pyramid~\cite{ref70} & ResNet-152 & VGGish & ImageNet & AudioSet & 44.0 & 176.3 & 77.8 \\
CMBS~\cite{ref71} & ResNet-152 & VGGish & ImageNet & AudioSet & 14.4 & 216.7 & 79.7 \\
\hline
\multicolumn{8}{l}{\textit{Shared-backbone transformer setting}} \\
LAVisH~\cite{ref12} & ViT-B-16 (shared) & -- & ImageNet & -- & 4.7 & 107.2 & 75.3 \\
AVMoE~\cite{ref13} & ViT-B-16 (shared) & -- & ImageNet & -- & 47.8 & 150.4 & 76.4 \\
\textbf{DPC-Net (Ours)} & ViT-B-16 (shared) & -- & ImageNet & -- & 5.9 & 108.4 & \textbf{77.8} \\
LAVisH~\cite{ref12} & ViT-L-16 (shared) & -- & ImageNet & -- & 14.5 & 340.1 & 78.1 \\
AVMoE~\cite{ref13} & ViT-L-16 (shared) & -- & ImageNet & -- & 147.7 & 483.1 & 79.2 \\
\textbf{DPC-Net (Ours)} & ViT-L-16 (shared) & -- & ImageNet & -- & 18.3 & 343.9 & \textbf{80.1} \\
LAVisH~\cite{ref12} & Swin-V2-B (shared) & -- & ImageNet & -- & 5.0 & 114.0 & 78.8 \\
AVMoE~\cite{ref13} & Swin-V2-B (shared) & -- & ImageNet & -- & 84.9 & 206.6 & 79.4 \\
\textbf{DPC-Net (Ours)} & Swin-V2-B (shared) & -- & ImageNet & -- & 6.6 & 115.6 & \textbf{80.7} \\
LAVisH~\cite{ref12} & Swin-V2-L (shared) & -- & ImageNet & -- & 10.1 & 238.8 & 81.1 \\
AVMoE\(\dagger\)~\cite{ref13} & Swin-V2-L (shared) & -- & ImageNet & -- & 147.5 & 347.4 & 81.0 \\
\textbf{DPC-Net (Ours)} & Swin-V2-L (shared) & -- & ImageNet & -- & 20.2 & 248.9 & \textbf{82.0} \\
\hline
\multicolumn{8}{l}{\textit{Separate audio-visual backbone setting}} \\
LAVisH~\cite{ref12} & Swin-V2-L & HTS-AT & ImageNet & AudioSet & 114.7 & 247.9 & 78.6 \\
DG-SCT\(\dagger\)~\cite{ref37} & Swin-V2-L & HTS-AT & ImageNet & AudioSet & 201.1 & 461.3 & 81.8 \\
AVMoE\(\dagger\)~\cite{ref13} & Swin-V2-L & HTS-AT & ImageNet & AudioSet & 141.0 & 404.0 & 82.1 \\
\textbf{DPC-Net (Ours)} & Swin-V2-L & HTS-AT & ImageNet & AudioSet & 125.8 & 393.2 & \textbf{83.3} \\
\hline
\end{tabular}}
\end{table*}

Table~\ref{tab:avel_results} reports segment-level localization results on AVE. DPC-Net improves accuracy across shared-backbone transformer settings and achieves the best result in the separate audio-visual backbone setting. The gain is obtained with fewer trainable parameters than AVMoE in matched settings, indicating that the improvement comes from readiness-guided bottleneck correction rather than a larger expert or routing capacity. Since AVEL uses segment-level categorical supervision, these results complement the AVSS reconstruction evaluation.

\subsection{Sequence-level Recognition: AVSR}
\label{sec:avsr_results}

\begin{table}[t]
\centering
\caption{
AVSR evaluation under varying acoustic conditions. Lower WER indicates better performance.
}
\label{tab:avsr_acoustic}
\setlength{\tabcolsep}{2.5pt}
\renewcommand{\arraystretch}{1.03}
\resizebox{\linewidth}{!}{%
\begin{tabular}{llcccccc}
\hline
Input & Method & \(-5\) dB & 0 dB & 5 dB & 10 dB & Clean & AVG \\
\hline
\multicolumn{8}{l}{\textit{LRS2}} \\
A & ASR~\cite{ref72} & 29.1 & 10.2 & 8.7 & 7.4 & 4.9 & 12.1 \\
A & AVEC~\cite{ref73} & 70.5 & 27.1 & 8.6 & 7.6 & 3.1 & 19.5 \\
AV & Conformer~\cite{ref72} & 24.9 & 16.5 & 10.8 & 7.8 & 4.6 & 12.9 \\
AV & V-CAFE~\cite{ref74} & 22.4 & 11.0 & 6.4 & 5.5 & 4.3 & 9.9 \\
AV & AVEC~\cite{ref73} & 9.7 & 5.0 & 3.4 & 2.8 & 2.6 & 4.7 \\
AV & AV-Relscore\(\dagger\)~\cite{ref53} & 10.8 & 5.9 & 4.9 & 4.0 & 3.9 & 5.9 \\
AV & A+VH~\cite{ref75} & 12.6 & 7.1 & 3.9 & 3.1 & 2.6 & 5.9 \\
AV & AD-AVSR~\cite{ref76} & 9.4 & 6.0 & 3.6 & 2.8 & 2.4 & 4.8 \\
AV & \textbf{DPC-Net (Ours)} & \textbf{9.0} & 5.8 & \textbf{3.3} & \textbf{2.6} & \textbf{2.3} & \textbf{4.6} \\
\hline
\multicolumn{8}{l}{\textit{LRS3}} \\
A & ASR~\cite{ref72} & -- & -- & -- & -- & -- & -- \\
A & AVEC~\cite{ref73} & 75.9 & 32.4 & 9.3 & 4.1 & 2.3 & 20.7 \\
AV & Conformer~\cite{ref72} & 22.3 & 14.6 & 8.3 & 5.4 & 3.2 & 10.8 \\
AV & V-CAFE~\cite{ref74} & 19.3 & 12.5 & 8.4 & 4.0 & 2.9 & 9.4 \\
AV & AV-Hubert~\cite{ref77} & 16.6 & 5.8 & 2.6 & 2.1 & 2.0 & 5.8 \\
AV & AVEC~\cite{ref73} & 11.2 & 4.9 & 3.1 & 2.5 & 2.0 & 4.7 \\
AV & AV-Relscore\(\dagger\)~\cite{ref53} & 8.3 & 4.6 & 3.0 & 2.7 & 2.6 & 4.2 \\
AV & A+VH~\cite{ref75} & 14.3 & 6.4 & 3.4 & 2.2 & 2.2 & 5.7 \\
AV & AD-AVSR~\cite{ref76} & 8.2 & 4.7 & 3.2 & 2.1 & \textbf{2.0} & 4.0 \\
AV & \textbf{DPC-Net (Ours)} & \textbf{7.8} & \textbf{4.5} & 3.3 & \textbf{2.0} & \textbf{2.0} & \textbf{3.9} \\
\hline
\end{tabular}}
\end{table}

Table~\ref{tab:avsr_acoustic} reports AVSR results under clean and noisy acoustic conditions. DPC-Net achieves the lowest average WER on both LRS2 and LRS3, with clear gains under severe acoustic noise. The downstream Conformer encoder, CTC branch, Transformer decoder, and recognition supervision are preserved, so the reduction in WER reflects improved audio-visual representations delivered by the encoder.

\begin{table}[t]
\centering
\caption{
AVSR robustness evaluation under joint audio-visual corruption on LRS2. Lower WER indicates better performance.
}
\label{tab:avsr_joint_corruption}
\setlength{\tabcolsep}{2.8pt}
\renewcommand{\arraystretch}{1.03}
\resizebox{\linewidth}{!}{%
\begin{tabular}{lcccccc}
\hline
Method & \(-5\) dB & 0 dB & 5 dB & 10 dB & Clean & AVG \\
\hline
\multicolumn{7}{l}{\textit{Occlusion}} \\
Conformer~\cite{ref72} & 25.1 & 16.6 & 10.8 & 8.0 & 4.9 & 13.1 \\
V-CAFE~\cite{ref74} & 22.4 & 11.3 & 6.5 & 5.7 & 4.4 & 10.1 \\
AV-Relscore~\cite{ref53} & 11.3 & 6.4 & 5.2 & 4.4 & 4.2 & 6.3 \\
AV-Relscore\(\dagger\)~\cite{ref53} & 10.8 & 6.1 & 5.0 & 4.1 & 4.0 & 6.0 \\
\textbf{DPC-Net (Ours)} & \textbf{9.1} & \textbf{6.1} & \textbf{3.7} & \textbf{2.8} & \textbf{2.4} & \textbf{4.8} \\
\hline
\multicolumn{7}{l}{\textit{Noise + Blur}} \\
Conformer~\cite{ref72} & 25.7 & 16.8 & 10.7 & 7.8 & 4.8 & 13.2 \\
V-CAFE~\cite{ref74} & 22.8 & 11.4 & 6.4 & 5.6 & 4.9 & 10.2 \\
AV-Relscore\(\dagger\)~\cite{ref53} & 10.7 & \textbf{6.1} & 4.9 & 4.0 & 4.2 & 6.1 \\
AD-AVSR~\cite{ref76} & 9.5 & 6.3 & 3.7 & 3.2 & 2.8 & 5.1 \\
\textbf{DPC-Net (Ours)} & \textbf{9.0} & 6.3 & \textbf{3.5} & \textbf{2.9} & \textbf{2.6} & \textbf{4.9} \\
\hline
\end{tabular}}
\end{table}

Table~\ref{tab:avsr_joint_corruption} evaluates AVSR under joint audio-visual corruption. DPC-Net achieves the lowest average WER under both occlusion and noise+blur, and remains best or competitive across acoustic conditions. These settings combine unreliable acoustic evidence with degraded visual cues, making them a strong test of whether readiness-guided correction can maintain useful encoder representations.

The AVSR results complement the AVSS and AVEL evaluations. AVSS tests signal-level reconstruction, AVEL tests segment-level classification, and AVSR tests sequence-level linguistic prediction. Across these output regimes, DPC-Net consistently acts on the shared encoder-level fusion process, supporting the generality of propagation-aware readiness control.

\subsection{Statistical Reliability}
\label{sec:statistical_reliability}

\begin{table}[t]
\centering
\caption{
Statistical reliability on representative controlled settings. Mean and standard deviation are reported over five runs.
}
\label{tab:statistical_reliability}
\setlength{\tabcolsep}{3.2pt}
\renewcommand{\arraystretch}{1.05}
\resizebox{\linewidth}{!}{%
\begin{tabular}{llccc}
\hline
Task & Setting & Compared Method & Compared Result & DPC-Net Result \\
\hline
AVSS & LRS2 Clean & AV-CrossNet & \(16.46\pm0.21\) & \(\mathbf{16.84\pm0.18}\) \\
AVSS & LRS2 Occlusion & RAVSS & \(13.17\pm0.16\) & \(\mathbf{14.82\pm0.13}\) \\
AVSS & LRS2 Noise+Blur & RAVSS & \(13.81\pm0.14\) & \(\mathbf{14.91\pm0.11}\) \\
AVEL & Swin-V2-L+HTS-AT & AVMoE & \(82.13\pm0.10\) & \(\mathbf{83.28\pm0.11}\) \\
AVSR & LRS2 \(-5\) dB & AD-AVSR & \(9.44\pm0.13\) & \(\mathbf{9.05\pm0.12}\) \\
\hline
\end{tabular}}
\end{table}

Table~\ref{tab:statistical_reliability} reports repeated-run results on representative controlled settings. DPC-Net maintains stable gains over the strongest compared methods across clean, visually degraded, and low-SNR conditions. The gains are especially clear under degraded visual input and low-SNR recognition, where unreliable evidence makes readiness assessment more important. These results indicate that the main improvements are stable across repeated runs. Across the three output regimes, the improvements show a consistent pattern. AVSS evaluates whether readiness-guided correction benefits signal-level reconstruction, AVEL evaluates whether it improves segment-level event reasoning, and AVSR evaluates whether it supports sequence-level linguistic prediction. Since DPC-Net is inserted into the audio-visual encoder in all tasks, these results indicate that the proposed mechanism improves the shared stage-wise representation process rather than relying on a task-specific output module. The stronger gains under visual degradation and low-SNR conditions further suggest that readiness control is especially useful when local cross-modal agreement becomes less stable and additional support evidence is needed for later propagation.

\section{Analysis of Representation Readiness}
\label{sec:mechanism_analysis}

This section evaluates whether the learned readiness-deficiency surrogate \(\widehat{D}_l\) behaves as an intervention-oriented bottleneck score. The analysis focuses on four questions: whether the two components of DPC-Net are necessary, whether the selected stage can be replaced by simpler criteria, whether the selected bottleneck is more sensitive and recoverable under intervention, and whether the internal readiness cues show interpretable trajectories under degraded input.

\subsection{Component and Selection-Criterion Analysis}
\label{sec:component_selection_analysis}

We first examine the contributions of commitment assessment (CA) and support-aware bottleneck completion (BC). Table~\ref{tab:component_analysis} reports results across AVSS, AVEL, and AVSR.

\begin{table}[t]
\centering
\caption{
Component analysis across AVSS, AVEL, and AVSR. CA and BC denote commitment assessment and support-aware bottleneck completion, respectively.
}
\label{tab:component_analysis}
\setlength{\tabcolsep}{4.5pt}
\renewcommand{\arraystretch}{1.05}
\resizebox{\linewidth}{!}{%
\begin{tabular}{lccccc c ccc}
\hline
\multirow{2}{*}{Variant} &
\multicolumn{2}{c}{Component} &
\multicolumn{3}{c}{AVSS SI-SNRi \(\uparrow\)} &
\multirow{2}{*}{AVEL Acc. \(\uparrow\)} &
\multicolumn{3}{c}{AVSR WER \(\downarrow\)} \\
\cline{2-6} \cline{8-10}
& CA & BC & Clean & Occ & N+B &  & Clean & Occ & N+B \\
\hline
\multicolumn{10}{l}{\textit{Main ablation of DPC-Net}} \\
DPC-Net (full)       & \checkmark & \checkmark & 16.8 & 14.8 & 14.9 & 83.3 & 4.6 & 4.8 & 4.9 \\
w/o CA               & \(\times\)  & \checkmark & 13.8 & 12.0 & 12.3 & 81.7 & 6.0 & 6.8 & 6.5 \\
w/o BC               & \checkmark & \(\times\)  & 14.4 & 12.9 & 12.5 & 82.2 & 5.7 & 6.4 & 6.6 \\
w/o CA and BC        & \(\times\)  & \(\times\)  & 12.8 & 9.9  & 10.1 & 80.4 & 9.5 & 10.8 & 10.5 \\
\hline
\multicolumn{10}{l}{\textit{Attention-only alternatives}} \\
Cross attention      & -- & -- & 13.4 & 11.4 & 11.6 & 81.9 & 8.2 & 9.3 & 9.0 \\
Pooling attention    & -- & -- & 14.1 & 12.3 & 12.4 & 82.1 & 7.5 & 8.6 & 8.3 \\
Channel attention    & -- & -- & 13.1 & 10.8 & 11.1 & 81.5 & 8.8 & 9.9 & 9.6 \\
\hline
\end{tabular}}
\end{table}

Table~\ref{tab:component_analysis} shows that CA and BC play complementary roles. Removing CA weakens performance because support correction loses the readiness-guided intervention target. Removing BC also reduces performance, showing that locating the vulnerable stage should be followed by support-aware correction. Attention-only alternatives also underperform the full model, indicating that the improvement comes from readiness-guided bottleneck localization and targeted correction rather than generic interaction enhancement. We further evaluate whether \(\widehat{D}_l\) can be replaced by simpler selection criteria. All variants in Fig.~\ref{fig:selection_criterion} use the same support-aware completion module and differ only in bottleneck selection.

\begin{figure}[t]
    \centering
    \includegraphics[width=\linewidth]{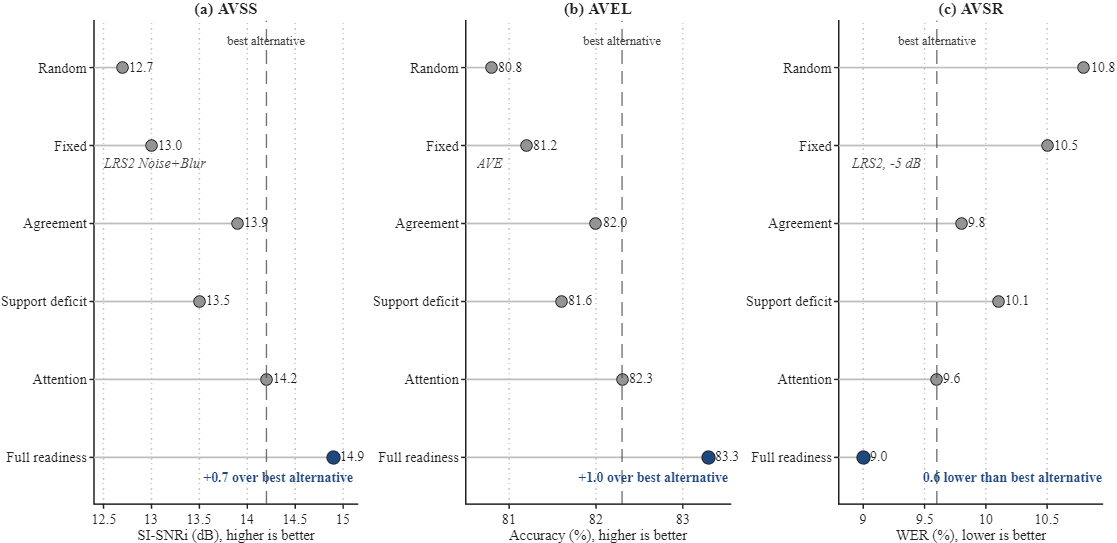}
    \caption{
    Selection-criterion validation. All variants use the same support-aware completion module and differ only in bottleneck selection. The dashed line denotes the best non-full criterion in each panel.
    }
    \label{fig:selection_criterion}
\end{figure}

Fig.~\ref{fig:selection_criterion} compares full readiness selection with random, fixed, agreement-only, support-deficit-only, and attention-response selection. Random and fixed selection remove input-dependent readiness estimation. Agreement-only selection focuses on current correspondence, support-deficit selection prioritizes weakly supported stages, and attention-response selection reflects interaction magnitude. The full readiness score performs best across the representative settings, reaching \(14.9\) dB SI-SNRi on AVSS under noise+blur, \(83.3\%\) accuracy on AVEL, and \(9.0\%\) WER on AVSR under \(-5\) dB acoustic noise. These results show that effective bottleneck localization requires the joint consideration of agreement, downstream anchoring, and support insufficiency.

\subsection{Counterfactual Intervention}
\label{sec:counterfactual_intervention}

We next examine whether the stage selected by \(\widehat{D}_l\) is more intervention-sensitive than control stages. We freeze the trained model and apply a matched-energy perturbation to three candidate targets: the selected bottleneck, its adjacent stage, and a random stage. For a candidate stage \(l\), the perturbation is
\begin{equation}
\Pi_l(z_l)
=
z_l
+
\epsilon
\frac{\|z_l\|_F}{\|\xi_l\|_F+\epsilon_0}
\xi_l,
\qquad
\xi_l \sim \mathcal{N}(0,I),
\label{eq:counterfactual_perturbation}
\end{equation}
where \(\epsilon=0.1\) and \(\epsilon_0=10^{-6}\). Degradation is measured as the decrease in SI-SNRi and accuracy for AVSS and AVEL, and as the increase in WER for AVSR. Recovery is measured as the performance regained after applying support-aware completion to the perturbed target.

\begin{figure}[t]
    \centering
    \includegraphics[width=\linewidth]{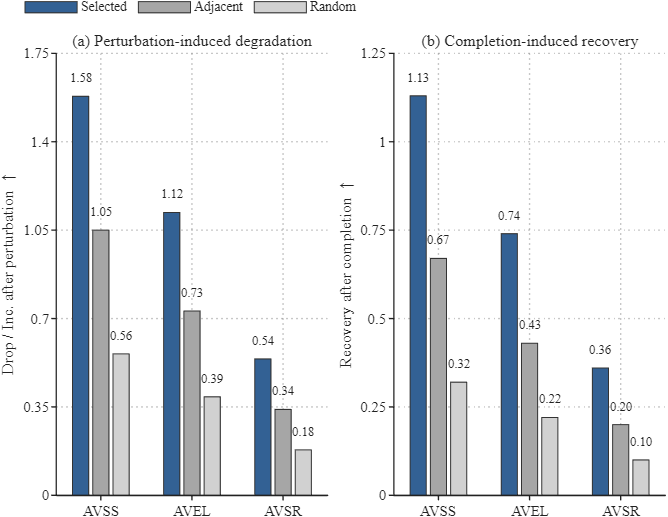}
    \caption{
    Counterfactual intervention on selected, adjacent, and random targets.
    \textbf{(a)} Perturbation-induced degradation.
    \textbf{(b)} Completion-induced recovery.
    Values are compared within each task because the metrics have different units and directions.
    }
    \label{fig:counterfactual_intervention}
\end{figure}

Fig.~\ref{fig:counterfactual_intervention} shows that the selected bottleneck is both more sensitive to perturbation and more responsive to completion. Perturbing the selected target causes the largest degradation, with a \(1.58\) dB SI-SNRi drop on AVSS, a \(1.12\) percentage-point accuracy drop on AVEL, and a \(0.54\) percentage-point WER increase on AVSR. Applying support-aware completion to the selected target also yields the strongest recovery, reaching \(1.13\) dB on AVSS, \(0.74\) percentage-point accuracy on AVEL, and \(0.36\) percentage-point WER on AVSR. This pattern supports the interpretation that \(\widehat{D}_l\) localizes a stage where intervention is more disruptive when perturbed and more beneficial when corrected.

\subsection{Readiness Trajectory under Degradation}
\label{sec:readiness_trajectory}

We further examine whether the internal cues of the readiness formulation show interpretable stage-wise behavior under degraded visual conditions. While the counterfactual analysis evaluates where intervention is most harmful and recoverable, this trajectory analysis examines how \(\widehat{A}_l\), \(\widehat{P}_l\), \(\widehat{C}_l\), and \(\widehat{D}_l\) evolve across stages and after support-aware completion.

\begin{figure*}[t]
    \centering
    \includegraphics[width=\textwidth]{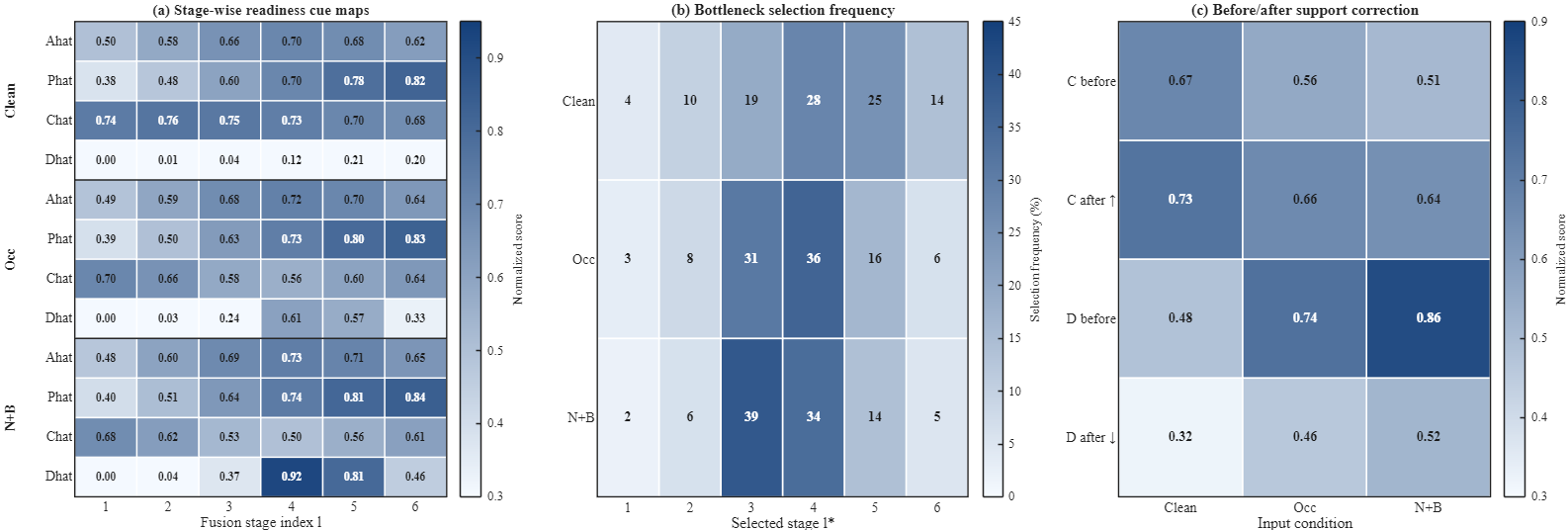}
    \caption{
    Readiness trajectory under visual degradation.
    \textbf{(a)} Stage-wise maps of \(\widehat{A}_l\), \(\widehat{P}_l\), \(\widehat{C}_l\), and \(\widehat{D}_l\).
    \textbf{(b)} Distribution of selected bottleneck stages.
    \textbf{(c)} Before/after support-aware completion at \(l^\star\).
    }
    \label{fig:readiness_trajectory}
\end{figure*}

Fig.~\ref{fig:readiness_trajectory}(a) shows the stage-wise maps of the four readiness-related quantities. Under clean input, \(\widehat{C}_l\) remains relatively stable and \(\widehat{D}_l\) stays low or moderate. Under occlusion and noise+blur, \(\widehat{C}_l\) decreases more clearly at intermediate stages, while \(\widehat{A}_l\) and \(\widehat{P}_l\) can remain high. This pattern is consistent with premature commitment: a stage may appear locally plausible and downstream-anchored even when its support coverage becomes insufficient, causing \(\widehat{D}_l\) to concentrate where agreement, downstream anchoring, and support insufficiency coexist.

Fig.~\ref{fig:readiness_trajectory}(b) shows that bottleneck selection is input-dependent. Under clean input, the selected stages are more broadly distributed. Under degraded visual conditions, the selections become more concentrated around intermediate fusion stages. Fig.~\ref{fig:readiness_trajectory}(c) further shows that support-aware correction increases \(\widehat{C}_{l^\star}\) and reduces \(\widehat{D}_{l^\star}\) at the selected stage. These results connect the internal trajectory of the readiness cues with the intervention behavior in Fig.~\ref{fig:counterfactual_intervention}.

\subsection{Implementation Design Analysis}
\label{sec:implementation_summary}

We further examine representative implementation choices of commitment assessment and support-aware bottleneck completion. The goal is to verify whether the practical design of DPC-Net follows the readiness-control principle, rather than relying on a single arbitrary implementation choice. Table~\ref{tab:implementation_design_compact} reports compact ablations on LRS2 under clean visuals, patch occlusion, and noise+blur.

\begin{table}[t]
\centering
\caption{
Compact implementation ablations on LRS2. SI-SNRi (dB, \(\uparrow\)) is reported under clean visuals, patch occlusion (Occ), and noise+blur (N+B). A and V denote audio-side and visual-side cues or support branches.
}
\label{tab:implementation_design_compact}
\setlength{\tabcolsep}{4.2pt}
\renewcommand{\arraystretch}{1.05}
\resizebox{\linewidth}{!}{%
\begin{tabular}{llccc}
\hline
Group & Variant & Clean & Occ & N+B \\
\hline
\multirow{3}{*}{Routing}
& Hard train + hard infer & 16.0 & 13.3 & 13.0 \\
& Soft train + soft infer & 16.3 & 14.4 & 14.2 \\
& Soft train + hard infer & \textbf{16.8} & \textbf{14.8} & 14.9 \\
\hline
\multirow{4}{*}{Assessment cue}
& Learned A + learned V & \textbf{16.8} & \textbf{14.8} & \textbf{14.9} \\
& Learned A + random V & 15.7 & 13.2 & 13.0 \\
& Random A + learned V & 15.3 & 12.9 & 12.7 \\
& Random A + random V & 12.5 & 9.0 & 9.2 \\
\hline
\multirow{3}{*}{Corrected stage}
& Learned single-stage & 16.8 & 14.8 & 14.9 \\
& Top-2 simultaneous & 16.9 & 14.6 & 14.8 \\
& Iter-2 sequential & \textbf{17.1} & \textbf{14.9} & \textbf{15.3} \\
\hline
\multirow{4}{*}{Support branch}
& Full A-support + full V-support & \textbf{16.8} & \textbf{14.8} & \textbf{14.9} \\
& V-support only & 15.9 & 13.5 & 13.5 \\
& A-support only & 16.0 & 13.4 & 13.6 \\
& No A/V support & 14.6 & 11.6 & 11.8 \\
\hline
\multirow{3}{*}{Support design}
& Full completion module & \textbf{16.8} & \textbf{14.8} & \textbf{14.9} \\
& w/o cross-modal interaction & 16.2 & 13.9 & 13.8 \\
& w/o global aggregation & 15.6 & 12.9 & 12.7 \\
\hline
\end{tabular}}
\end{table}

Table~\ref{tab:implementation_design_compact} shows that differentiable routing during training and hard bottleneck selection during inference provide the strongest default configuration. The assessment-cue ablations further show that readiness localization benefits from input-dependent evidence from both modalities. Replacing either side with random cues reduces performance, and replacing both sides causes a much larger drop. Iterative two-stage correction can provide slightly higher SI-SNRi, especially under noise+blur, but it requires an additional assessment-completion pass. We therefore use learned single-stage intervention as the default design for a better balance between effectiveness and efficiency.

The support-branch and support-design ablations show that the selected bottleneck benefits from both audio-side and visual-side support. One-sided support remains useful, while removing both support branches produces a large degradation. Removing cross-modal support interaction or global support aggregation also weakens performance, indicating that the correction should be constructed from refined all-stage audio-visual support. These results support the implementation choice of DPC-Net: estimate an intervention-sensitive bottleneck and correct it through targeted, support-aware residual updating.

\section{Conclusion}
\label{sec:conclusion}

This work studied stage-wise audio-visual learning from the perspective of propagation-aware representation readiness. We formulated premature perceptual commitment as a readiness-deficiency problem, where local plausibility, propagation influence, and support insufficiency jointly appear at an intermediate fusion stage. Based on this formulation, we proposed DPC-Net, an encoder-level readiness-control framework that estimates an observable readiness-deficiency surrogate, localizes the intervention-sensitive bottleneck, and applies support-aware correction using cross-layer and cross-modal evidence.

Experiments across AVSS, AVEL, and AVSR show that DPC-Net improves audio-visual representation learning under reconstruction, localization, and recognition regimes while preserving task-specific heads, losses, and decoding modules. Mechanism-oriented analyses further show that the learned surrogate selects more effective bottleneck stages than simpler criteria, identifies representations that are sensitive to perturbation and responsive to completion, and produces interpretable readiness trajectories under degraded visual conditions. These results support readiness-guided bottleneck correction as a useful principle for stage-wise audio-visual fusion.

\balance
\bibliographystyle{IEEEbib}
\bibliography{refs}

@article{ref1,
  title={Deep audio-visual learning: A survey},
  author={Zhu, Hao and Luo, Man-Di and Wang, Rui and Zheng, Ai-Hua and He, Ran},
  journal={International Journal of Automation and Computing},
  volume={18},
  number={3},
  pages={351--376},
  year={2021},
  publisher={Springer}
}

@article{ref2,
  title={Deep Learning for Visual Speech Analysis: A Survey},
  author={Sheng, Changchong and Kuang, Gangyao and Bai, Liang and Hou, Chenping and Guo, Yulan and Xu, Xin and Pietikainen, Matti and Liu, Li},
  journal={IEEE Transactions on Pattern Analysis and Machine Intelligence},
  volume={46},
  number={09},
  pages={6001--6022},
  year={2024},
  publisher={IEEE Computer Society}
}

@article{ref3,
  title   = {Multimodal Alignment and Fusion: A Survey},
  author  = {Li, Songtao and Tang, Hao},
  journal = {International Journal of Computer Vision},
  volume  = {134},
  pages   = {103},
  year    = {2026},
  doi     = {10.1007/s11263-025-02667-1}
}

@article{ref4,
  title={Deep audio-visual speech recognition},
  author={Afouras, T and Chung, J and Senior, A and Vinyals, O and Zisserman, A},
  journal={IEEE Transactions on Pattern Analysis and Machine Intelligence},
  volume={44},
  number={12},
  year={2018},
  publisher={IEEE}
}

@inproceedings{ref5,
  title={Looking into your speech: Learning cross-modal affinity for audio-visual speech separation},
  author={Lee, Jiyoung and Chung, Soo-Whan and Kim, Sunok and Kang, Hong-Goo and Sohn, Kwanghoon},
  booktitle={Proceedings of the IEEE/CVF Conference on Computer Vision and Pattern Recognition},
  pages={1336--1345},
  year={2021}
}

@inproceedings{ref6,
  title={Audio-visual event localization via recursive fusion by joint co-attention},
  author={Duan, Bin and Tang, Hao and Wang, Wei and Zong, Ziliang and Yang, Guowei and Yan, Yan},
  booktitle={Proceedings of the IEEE/CVF winter conference on applications of computer vision},
  pages={4013--4022},
  year={2021}
}

@inproceedings{ref7,
  title={Mlca-avsr: Multi-layer cross attention fusion based audio-visual speech recognition},
  author={Wang, He and Guo, Pengcheng and Zhou, Pan and Xie, Lei},
  booktitle={ICASSP 2024-2024 IEEE International Conference on Acoustics, Speech and Signal Processing (ICASSP)},
  pages={8150--8154},
  year={2024},
  organization={IEEE}
}

@inproceedings{ref8,
  title={Learning event-specific localization preferences for audio-visual event localization},
  author={Ge, Shiping and Jiang, Zhiwei and Yin, Yafeng and Wang, Cong and Cheng, Zifeng and Gu, Qing},
  booktitle={Proceedings of the 31st ACM International Conference on Multimedia},
  pages={3446--3454},
  year={2023}
}

@article{ref10,
  title={Unified cross-modal attention: robust audio-visual speech recognition and beyond},
  author={Li, Jiahong and Li, Chenda and Wu, Yifei and Qian, Yanmin},
  journal={IEEE/ACM Transactions on Audio, Speech, and Language Processing},
  volume={32},
  pages={1941--1953},
  year={2024},
  publisher={IEEE}
}

@inproceedings{ref11,
  title={IIANet: an intra-and inter-modality attention network for audio-visual speech separation},
  author={Li, Kai and Yang, Runxuan and Sun, Fuchun and Hu, Xiaolin},
  booktitle={Proceedings of the 41st International Conference on Machine Learning},
  pages={29181--29200},
  year={2024}
}

@inproceedings{ref12,
  title={Vision transformers are parameter-efficient audio-visual learners},
  author={Lin, Yan-Bo and Sung, Yi-Lin and Lei, Jie and Bansal, Mohit and Bertasius, Gedas},
  booktitle={Proceedings of the IEEE/CVF Conference on Computer Vision and Pattern Recognition},
  pages={2299--2309},
  year={2023}
}

@article{ref13,
  title={Mixtures of experts for audio-visual learning},
  author={Cheng, Ying and Li, Yang and He, Junjie and Feng, Rui},
  journal={Advances in Neural Information Processing Systems},
  volume={37},
  pages={219--243},
  year={2024}
}

@inproceedings{ref14,
  title={Robust audio-visual segmentation via audio-guided visual convergent alignment},
  author={Liu, Chen and Li, Peike and Yang, Liying and Wang, Dadong and Li, Lincheng and Yu, Xin},
  booktitle={Proceedings of the Computer Vision and Pattern Recognition Conference},
  pages={28922--28931},
  year={2025}
}

@inproceedings{ref15,
  title={Rethink cross-modal fusion in weakly-supervised audio-visual video parsing},
  author={Xu, Yating and Hu, Conghui and Lee, Gim Hee},
  booktitle={Proceedings of the IEEE/CVF Winter Conference on Applications of Computer Vision},
  pages={5615--5624},
  year={2024}
}

@inproceedings{ref16,
  title={Positive sample propagation along the audio-visual event line},
  author={Zhou, Jinxing and Zheng, Liang and Zhong, Yiran and Hao, Shijie and Wang, Meng},
  booktitle={Proceedings of the IEEE/CVF conference on computer vision and pattern recognition},
  pages={8436--8444},
  year={2021}
}

@inproceedings{ref17,
  title={Dynamic cross attention for audio-visual person verification},
  author={Praveen, R Gnana and Alam, Jahangir},
  booktitle={2024 IEEE 18th International Conference on Automatic Face and Gesture Recognition (FG)},
  pages={1--5},
  year={2024},
  organization={IEEE}
}

@article{ref18,
  title={Avs-mamba: Exploring temporal and multi-modal mamba for audio-visual segmentation},
  author={Gong, Sitong and Zhuge, Yunzhi and Zhang, Lu and Wang, Yifan and Zhang, Pingping and Wang, Lijun and Lu, Huchuan},
  journal={IEEE Transactions on Multimedia},
  year={2025},
  publisher={IEEE}
}

@inproceedings{ref19,
  title={Progressive Homeostatic and Plastic Prompt Tuning for Audio-Visual Multi-Task Incremental Learning},
  author={Yin, Jiong and Li, Liang and Zhang, Jiehua and Gao, Yuhan and Yan, Chenggang and Sheng, Xichun},
  booktitle={Proceedings of the IEEE/CVF International Conference on Computer Vision},
  pages={2022--2033},
  year={2025}
}

@article{ref22,
  title={Contribution-aware Dynamic Multi-modal Balance for Audio-Visual Speech Separation},
  author={Xu, Xinmeng and Tu, Weiping and Yang, Yuhong and Li, Jizhen and Zhang, Yiqun and Chen, Hongyang},
  journal={IEEE Transactions on Multimedia},
  year={2026},
  publisher={IEEE}
}

@article{ref25,
  title={Efficient audio--visual information fusion using encoding pace synchronization for Audio--Visual Speech Separation},
  author={Xu, Xinmeng and Tu, Weiping and Yang, Yuhong},
  journal={Information Fusion},
  volume={115},
  pages={102749},
  year={2025},
  publisher={Elsevier}
}

@article{ref28,
  title={Bayesian causal inference: A unifying neuroscience theory},
  author={Shams, Ladan and Beierholm, Ulrik},
  journal={Neuroscience \& Biobehavioral Reviews},
  volume={137},
  pages={104619},
  year={2022},
  publisher={Elsevier}
}

@article{ref29,
  title={Older adults preserve audiovisual integration through enhanced cortical activations, not by recruiting new regions},
  author={Jones, Samuel A and Noppeney, Uta},
  journal={PLoS Biology},
  volume={22},
  number={2},
  pages={e3002494},
  year={2024},
  publisher={Public Library of Science San Francisco, CA USA}
}

@article{ref30,
  title={Causal inference in multisensory perception},
  author={K{\"o}rding, Konrad P and Beierholm, Ulrik and Ma, Wei Ji and Quartz, Steven and Tenenbaum, Joshua B and Shams, Ladan},
  journal={PLoS one},
  volume={2},
  number={9},
  pages={e943},
  year={2007},
  publisher={Public Library of Science San Francisco, USA}
}

@article{ref31,
  title={Cortical hierarchies perform Bayesian causal inference in multisensory perception},
  author={Rohe, Tim and Noppeney, Uta},
  journal={PLoS biology},
  volume={13},
  number={2},
  pages={e1002073},
  year={2015},
  publisher={Public Library of Science San Francisco, CA USA}
}

@article{ref32,
  title={Neural processing of asynchronous audiovisual speech perception},
  author={Stevenson, Ryan A and Altieri, Nicholas A and Kim, Sunah and Pisoni, David B and James, Thomas W},
  journal={Neuroimage},
  volume={49},
  number={4},
  pages={3308--3318},
  year={2010},
  publisher={Elsevier}
}

@article{ref33,
  title={The neural dynamics of hierarchical Bayesian causal inference in multisensory perception},
  author={Rohe, Tim and Ehlis, Ann-Christine and Noppeney, Uta},
  journal={Nature communications},
  volume={10},
  number={1},
  pages={1907},
  year={2019},
  publisher={Nature Publishing Group UK London}
}

@article{ref34,
  title={A causal inference model explains perception of the McGurk effect and other incongruent audiovisual speech},
  author={Magnotti, John F and Beauchamp, Michael S},
  journal={PLoS computational biology},
  volume={13},
  number={2},
  pages={e1005229},
  year={2017},
  publisher={Public Library of Science San Francisco, CA USA}
}

@article{ref36,
  title={The role of conflict processing in multisensory perception: behavioural and electroencephalography evidence},
  author={Marly, Adri{\`a} and Yazdjian, Arek and Soto-Faraco, Salvador},
  journal={Philosophical Transactions of the Royal Society B: Biological Sciences},
  volume={378},
  number={1886},
  year={2023},
  publisher={The Royal Society}
}

@article{ref37,
  title={Causal inference in the multisensory brain},
  author={Cao, Yinan and Summerfield, Christopher and Park, Hame and Giordano, Bruno Lucio and Kayser, Christoph},
  journal={Neuron},
  volume={102},
  number={5},
  pages={1076--1087},
  year={2019},
  publisher={Elsevier}
}

@inproceedings{ref38,
  title={Audio-visual event localization in unconstrained videos},
  author={Tian, Yapeng and Shi, Jing and Li, Bochen and Duan, Zhiyao and Xu, Chenliang},
  booktitle={Proceedings of the European conference on computer vision (ECCV)},
  pages={247--263},
  year={2018}
}

@inproceedings{ref39,
  title={Dual Attention Matching for Audio-Visual Event Localization},
  author={Wu, Yu and Zhu, Linchao and Yan, Yan and Yang, Yi},
  booktitle={2019 IEEE/CVF International Conference on Computer Vision (ICCV)},
  pages={6291--6299},
  year={2019},
  organization={IEEE Computer Society}
}

@article{ref40,
  title={Looking to listen at the cocktail party: a speaker-independent audio-visual model for speech separation},
  author={Ephrat, Ariel and Mosseri, Inbar and Lang, Oran and Dekel, Tali and Wilson, Kevin and Hassidim, Avinatan and Freeman, William T and Rubinstein, Michael},
  journal={ACM Transactions on Graphics (TOG)},
  volume={37},
  number={4},
  pages={1--11},
  year={2018},
  publisher={ACM New York, NY, USA}
}

@inproceedings{ref41,
  title={End-to-End Audiovisual Speech Recognition},
  author={Petridis, Stavros and Stafylakis, Themos and Ma, Pingehuan and Cai, Feipeng and Tzimiropoulos, Georgios and Pantic, Maja},
  booktitle={ICASSP 2018-2018 IEEE International Conference on Acoustics, Speech and Signal Processing (ICASSP)},
  pages={6548--6552},
  year={2018},
  organization={IEEE}
}

@inproceedings{ref42,
  author    = {Wang, Kai and Mo, Shentong and Tian, Yapeng and Hatzinakos, Dimitrios},
  title     = {Prompt Image to Watch and Hear: Multimodal Prompting for Parameter-Efficient Audio-Visual Learning},
  booktitle = {36th British Machine Vision Conference (BMVC)},
  year      = {2025},
  publisher = {BMVA},
  url       = {https://bmva-archive.org.uk/bmvc/2025/assets/papers/Paper_949/paper.pdf}
}

@inproceedings{ref43,
  title={FaceFilter: Audio-Visual Speech Separation Using Still Images},
  author={Chung, Soo-Whan and Choe, Soyeon and Chung, Joon Son and Kang, Hong-Goo},
  booktitle={Proc. Interspeech 2020},
  pages={3481--3485},
  year={2020}
}

@inproceedings{ref44,
  title={VisualVoice: Audio-Visual Speech Separation With Cross-Modal Consistency},
  author={Gao, Ruohan and Grauman, Kristen},
  booktitle={Proceedings of the IEEE/CVF Conference on Computer Vision and Pattern Recognition},
  pages={15495--15505},
  year={2021}
}

@Article{ref45,
AUTHOR = {Yu, Wentao and Zeiler, Steffen and Kolossa, Dorothea},
TITLE = {Reliability-Based Large-Vocabulary Audio-Visual Speech Recognition},
JOURNAL = {Sensors},
VOLUME = {22},
YEAR = {2022},
NUMBER = {15},
ARTICLE-NUMBER = {5501},
URL = {https://www.mdpi.com/1424-8220/22/15/5501},
PubMedID = {35898005},
ISSN = {1424-8220}
}

@ARTICLE{ref46,
  author={Wu, Wenxuan and Chen, Xueyuan and Wang, Shuai and Wang, Jiadong and Meng, Lingwei and Wu, Xixin and Meng, Helen and Li, Haizhou},
  journal={IEEE Journal of Selected Topics in Signal Processing}, 
  title={$C^{2}$AV-TSE: Context and Confidence-Aware Audio Visual Target Speaker Extraction}, 
  year={2025},
  volume={19},
  number={4},
  pages={646-657},
  doi={10.1109/JSTSP.2025.3560513}}

@inproceedings{ref47,
 author = {Duan, Haoyi and Xia, Yan and Mingze, Zhou and Tang, Li and Zhu, Jieming and Zhao, Zhou},
 booktitle = {Advances in Neural Information Processing Systems},
 editor = {A. Oh and T. Naumann and A. Globerson and K. Saenko and M. Hardt and S. Levine},
 pages = {56075--56094},
 publisher = {Curran Associates, Inc.},
 title = {Cross-modal Prompts: Adapting Large Pre-trained Models for Audio-Visual Downstream Tasks},
 url = {https://proceedings.neurips.cc/paper_files/paper/2023/file/af01716e08073368a7c8a62be46dba17-Paper-Conference.pdf},
 volume = {36},
 year = {2023}
}

@InProceedings{ref48,
    author    = {Wang, Kai and Tian, Yapeng and Hatzinakos, Dimitrios},
    title     = {Towards Efficient Audio-Visual Learners via Empowering Pre-trained Vision Transformers with Cross-Modal Adaptation},
    booktitle = {Proceedings of the IEEE/CVF Conference on Computer Vision and Pattern Recognition (CVPR) Workshops},
    month     = {June},
    year      = {2024},
    pages     = {1837-1846}
}

@article{ref49,
author = {Noppeney, Uta and Lee, Hwee Ling},
title = {Causal inference and temporal predictions in audiovisual perception of speech and music},
journal = {Annals of the New York Academy of Sciences},
volume = {1423},
number = {1},
pages = {102-116},
doi = {https://doi.org/10.1111/nyas.13615},
url = {https://nyaspubs.onlinelibrary.wiley.com/doi/abs/10.1111/nyas.13615},
eprint = {https://nyaspubs.onlinelibrary.wiley.com/doi/pdf/10.1111/nyas.13615},
year = {2018}
}

@article{ref50,
  title={Reliability-Weighted Integration of Audiovisual Signals Can Be Modulated by Top-down Attention},
  author={Noppeney, Uta and Rohe, Tim},
  journal={eNeuro},
  volume={5},
  number={1},
  pages={e0315--17},
  year={2018},
  publisher={Society for Neuroscience}
}

@article{ref51,
  title={Explicit estimation of magnitude and phase spectra in parallel for high-quality speech enhancement},
  author={Lu, Ye-Xin and Ai, Yang and Ling, Zhen-Hua},
  journal={Neural Networks},
  volume={189},
  pages={107562},
  year={2025},
  publisher={Elsevier}
}

@article{ref53,
  title={{LRS3-TED}: a large-scale dataset for visual speech recognition},
  author={Afouras, Triantafyllos and Chung, Joon Son and Zisserman, Andrew},
  journal={arXiv preprint arXiv: 1809.00496},
  year={2018}
}

@inproceedings{ref54,
  title     = {VoxCeleb2: Deep Speaker Recognition},
  author    = {Joon Son Chung and Arsha Nagrani and Andrew Zisserman},
  year      = {2018},
  booktitle = {Interspeech 2018},
  pages     = {1086--1090},
  doi       = {10.21437/Interspeech.2018-1929},
  issn      = {2958-1796},
}

@ARTICLE{ref56,
  author={Kalkhorani, Vahid Ahmadi and Yu, Cheng and Kumar, Anurag and Tan, Ke and Xu, Buye and Wang, DeLiang},
  journal={IEEE Journal of Selected Topics in Signal Processing}, 
  title={AV-CrossNet: An Audiovisual Complex Spectral Mapping Network for Speech Separation by Leveraging Narrow- and Cross-Band Modeling}, 
  year={2025},
  volume={19},
  number={4},
  pages={685-694},
  doi={10.1109/JSTSP.2025.3567838}}

@INPROCEEDINGS{ref62,
  author={Gemmeke, Jort F. and Ellis, Daniel P. W. and Freedman, Dylan and Jansen, Aren and Lawrence, Wade and Moore, R. Channing and Plakal, Manoj and Ritter, Marvin},
  booktitle={2017 IEEE International Conference on Acoustics, Speech and Signal Processing (ICASSP)}, 
  title={Audio Set: An ontology and human-labeled dataset for audio events}, 
  year={2017},
  volume={},
  number={},
  pages={776-780},
  doi={10.1109/ICASSP.2017.7952261}}

@InProceedings{ref63,
    author    = {Xia, Yan and Zhao, Zhou},
    title     = {Cross-Modal Background Suppression for Audio-Visual Event Localization},
    booktitle = {Proceedings of the IEEE/CVF Conference on Computer Vision and Pattern Recognition (CVPR)},
    month     = {June},
    year      = {2022},
    pages     = {19989-19998}
}

@ARTICLE{ref64,
  author={Li, Jiahong and Li, Chenda and Wu, Yifei and Qian, Yanmin},
  journal={IEEE/ACM Transactions on Audio, Speech, and Language Processing}, 
  title={Unified Cross-Modal Attention: Robust Audio-Visual Speech Recognition and Beyond}, 
  year={2024},
  volume={32},
  number={},
  pages={1941-1953},
  doi={10.1109/TASLP.2024.3375641}}

@inproceedings{ref65,
  title     = {Combining Residual Networks with LSTMs for Lipreading},
  author    = {Themos Stafylakis and Georgios Tzimiropoulos},
  year      = {2017},
  booktitle = {Interspeech 2017},
  pages     = {3652--3656},
  doi       = {10.21437/Interspeech.2017-85},
  issn      = {2958-1796},
}

@InProceedings{ref66,
author = {He, Kaiming and Zhang, Xiangyu and Ren, Shaoqing and Sun, Jian},
title = {Deep Residual Learning for Image Recognition},
booktitle = {Proceedings of the IEEE Conference on Computer Vision and Pattern Recognition (CVPR)},
month = {June},
year = {2016}
}

@InProceedings{ref67,
  title = 	 {Branchformer: Parallel {MLP}-Attention Architectures to Capture Local and Global Context for Speech Recognition and Understanding},
  author =       {Peng, Yifan and Dalmia, Siddharth and Lane, Ian and Watanabe, Shinji},
  booktitle = 	 {Proceedings of the 39th International Conference on Machine Learning},
  pages = 	 {17627--17643},
  year = 	 {2022},
  editor = 	 {Chaudhuri, Kamalika and Jegelka, Stefanie and Song, Le and Szepesvari, Csaba and Niu, Gang and Sabato, Sivan},
  volume = 	 {162},
  series = 	 {Proceedings of Machine Learning Research},
  month = 	 {17--23 Jul},
  publisher =    {PMLR}
}

@INPROCEEDINGS{ref68,
  author={Wang, He and Guo, Pengcheng and Zhou, Pan and Xie, Lei},
  booktitle={ICASSP 2024 - 2024 IEEE International Conference on Acoustics, Speech and Signal Processing (ICASSP)}, 
  title={MLCA-AVSR: Multi-Layer Cross Attention Fusion Based Audio-Visual Speech Recognition}, 
  year={2024},
  volume={},
  number={},
  pages={8150-8154},
  doi={10.1109/ICASSP48485.2024.10446769}}

@INPROCEEDINGS{ref69,
  author={Guo, Pengcheng and Wang, He and Mu, Bingshen and Zhang, Ao and Chen, Peikun},
  booktitle={ICASSP 2023 - 2023 IEEE International Conference on Acoustics, Speech and Signal Processing (ICASSP)}, 
  title={The NPU-ASLP System for Audio-Visual Speech Recognition in MISP 2022 Challenge}, 
  year={2023},
  volume={},
  number={},
  pages={1-2},
  doi={10.1109/ICASSP49357.2023.10097011}}

@article{ref70,
  title={Adam: A method for stochastic optimization},
  author={Kingma, Diederik P and Ba, Jimmy},
  journal={arXiv preprint arXiv:1412.6980},
  year={2014}
}

@INPROCEEDINGS{ref71,
  author={Wu, Jian and Xu, Yong and Zhang, Shi-Xiong and Chen, Lian-Wu and Yu, Meng and Xie, Lei and Yu, Dong},
  booktitle={2019 IEEE Automatic Speech Recognition and Understanding Workshop (ASRU)}, 
  title={Time Domain Audio Visual Speech Separation}, 
  year={2019},
  volume={},
  number={},
  pages={667-673},
  doi={10.1109/ASRU46091.2019.9003983}}

@inproceedings{ref72,
  title     = {Audio-Visual Speech Separation in Noisy Environments with a Lightweight Iterative Model},
  author    = {Héctor Martel and Julius Richter and Kai Li and Xiaolin Hu and Timo Gerkmann},
  year      = {2023},
  booktitle = {Interspeech 2023},
  pages     = {1673--1677},
  doi       = {10.21437/Interspeech.2023-1753},
  issn      = {2958-1796},
}

@inproceedings{ref73,
  title={{RTFS-Net: Recurrent Time-Frequency Modelling for Efficient Audio-Visual Speech Separation}},
  author={Pegg, Samuel and Li, Kai and Hu, Xiaolin},
  booktitle={The Twelfth International Conference on Learning Representations},
  year={2024}
}

@INPROCEEDINGS{ref74,
  author={Wu, Yifei and Li, Chenda and Bai, Jinfeng and Wu, Zhongqin and Qian, Yanmin},
  booktitle={ICASSP 2022 - 2022 IEEE International Conference on Acoustics, Speech and Signal Processing (ICASSP)}, 
  title={Time-Domain Audio-Visual Speech Separation on Low Quality Videos}, 
  year={2022},
  volume={},
  number={},
  pages={256-260},
  doi={10.1109/ICASSP43922.2022.9746866}}

@inproceedings{ref75,
  author={Xinmeng Xu and Yang Wang and Jie Jia and Binbin Chen and Dejun Li},
  title={{Improving Visual Speech Enhancement Network by Learning Audio-visual Affinity with Multi-head Attention}},
  year=2022,
  booktitle={Proc. Interspeech 2022},
  pages={971--975},
  doi={10.21437/Interspeech.2022-10041}
}

@inproceedings{ref76,
author = {Pan, Tianrui and Liu, Jie and Wang, Bohan and Tang, Jie and Wu, Gangshan},
title = {RAVSS: Robust Audio-Visual Speech Separation in Multi-Speaker Scenarios with Missing Visual Cues},
year = {2024},
isbn = {9798400706868},
publisher = {Association for Computing Machinery},
address = {New York, NY, USA},
url = {https://doi.org/10.1145/3664647.3681261},
doi = {10.1145/3664647.3681261},
booktitle = {Proceedings of the 32nd ACM International Conference on Multimedia},
pages = {4748–4756},
numpages = {9},
location = {Melbourne VIC, Australia},
series = {MM '24}
}

@inproceedings{ref77,
  title={Dual perspective network for audio-visual event localization},
  author={Rao, Varshanth and Khalil, Md Ibrahim and Li, Haoda and Dai, Peng and Lu, Juwei},
  booktitle={European Conference on Computer Vision},
  pages={689--704},
  year={2022},
  organization={Springer}
}

\vfill

\end{document}